\newif\ifSubmission
\Submissiontrue

\newif\ifComments
\Commentsfalse

\newif\ifAnonymous
\Anonymousfalse

\documentclass[12pt]{article}

\usepackage[letterpaper,margin=1.in]{geometry}
\usepackage[utf8]{inputenc}

\ifAnonymous
    \author{}
\else
    \author{Dakshita Khurana\thanks{NTT Research and University of Illinois Urbana-Champaign. Email: \texttt{dakshita@illinois.edu}} \quad Bhaskar Roberts\thanks{UC Berkeley. Work done in part at NTT Research. Email: \texttt{bhaskarroberts@gmail.com}} \quad Avishay Tal\thanks{UC Berkeley. Email: \texttt{atal@berkeley.edu}
    }}
\fi
\date{}

\usepackage{libertine}

\usepackage{amsthm}
\usepackage[T1]{fontenc}

\usepackage{enumitem}
\usepackage{graphicx} 
\usepackage{physics}
\usepackage{breakcites}
\usepackage{amsmath,bm}
\usepackage{amssymb}
\usepackage[libertine]{newtxmath}
\usepackage[colorlinks=true, allcolors=blue]{hyperref}
\usepackage[dvipsnames]{xcolor}
\usepackage{bbm}
\usepackage{cleveref}
\usepackage{mathtools}
\usepackage{algorithm, algpseudocode, float, algorithmicx}
\usepackage{mdframed}
\usepackage{caption}

\usepackage{multirow}

\usepackage[english]{babel}

\allowdisplaybreaks

\theoremstyle{definition}
\newtheorem{definition}{Definition}[section]
\newtheorem{theorem}{Theorem}[section]
\newtheorem{lemma}{Lemma}[section]

\newtheorem{claim}{Claim}[section]
\theoremstyle{definition}

\newcommand{\Prove}{\mathsf{Prove}}
\newcommand{\Verify}{\mathsf{Verify}}
\newcommand{\poly}{\mathsf{poly}}
\newcommand{\RO}{H}
\newcommand{\ro}{h}
\newcommand{\minent}{{h_\infty}}
\newcommand{\Minent}{{H_\infty}}

\newcommand{\secp}{\lambda}
\newcommand{\negl}{\mathsf{negl}}
\newcommand{\getsr}{\overset{\$}{\leftarrow}}
\newcommand{\bit}{\{0,1\}}

\newcommand{\bbC}{\mathbb{C}}

\newcommand{\bbF}{\mathbb{F}}

\newcommand{\bbN}{\mathbb{N}}
\newcommand{\bbR}{\mathbb{R}}

\newcommand{\bfX}{\mathbf{X}}
\newcommand{\bfx}{\mathbf{x}}

\newcommand{\bfz}{\mathbf{z}}

\newcommand{\cA}{\mathcal{A}}

\newcommand{\cH}{\mathcal{H}}
\newcommand{\cO}{\mathcal{O}}

\ifComments
    \newcommand{\authnote}[3]{\textcolor{#3}{[{\footnotesize {\bf #1:} { {#2}}}]}}
\else
    \newcommand{\authnote}[3]{}
\fi

\newcommand{\dakshita}[1]{\authnote{Dakshita}{#1}{Red}}

\title{Certified Randomness without Structure\\ Against Shallow-Query Adversaries}

\begin{document}

\maketitle
\begin{abstract}
In a recent breakthrough, Yamakawa and Zhandry (J. ACM 2024) constructed a proof of quantumness in the quantum random oracle model (QROM) in which the quantum prover samples a codeword preimage of a publicly computable function $H$. They conjectured that given any $H$, a successful prover must sample their preimage from a high-entropy distribution over possible answers. If true, this would give a certifiable randomness protocol in the quantum random oracle model. As partial evidence for their conjecture, Yamakawa and Zhandry proved the security of their certifiable randomness protocol assuming the Aaronson-Ambainis conjecture. 

We prove the security of the certifiable randomness protocol of Yamakawa-Zhandry unconditionally, without relying on the unproven Aaronson-Ambainis conjecture, against low query-depth quantum adversaries: specifically, adversaries that make up to $o(\log \lambda)$ adaptive quantum queries to the random oracle.
\end{abstract}
\thispagestyle{empty}
\newpage
\tableofcontents
\thispagestyle{empty}

\newpage
\setcounter{page}{1}

\section{Introduction}
Randomness is a fundamental resource in cryptography and computing, yet
generating \emph{certifiably random} outputs remains a significant challenge. Classical computers
cannot produce such randomness on their own, since any deterministic process is
in principle predictable. Quantum mechanics offers a potential solution: the
inherent probabilism of quantum measurement suggests that quantum devices could
generate outputs that are information-theoretically random. But how can a
classical verifier trust that a quantum prover is genuinely exploiting quantum
randomness, rather than outputting a biased or pre-determined answer?

\paragraph{The landscape of certified randomness.}
The \emph{device-independent} approach, pioneered by Colbeck~\cite{Col06}, Pironio et al.~\cite{PAM+10}, and Vazirani-Vidick~\cite{VV12}, achieves certified randomness by exploiting Bell inequality violations: if a pair of non-communicating quantum devices produces correlations that violate a Bell inequality, the outputs must contain genuine randomness. This approach offers information-theoretic security, but requires stringent experimental conditions, namely two spatially separated non-communicating devices sharing high-fidelity entangled states, that make deployment challenging.
 
An alternative line of work achieves certified randomness using \emph{computational assumptions}. Brakerski et al.~\cite{BCM+21} and the follow-up by Mahadev, Vazirani, and Vidick~\cite{MVV22}, showed how to certify randomness from a single quantum device under the hardness of the Learning With Errors (LWE) problem.
Subsequently, Aaronson and Hung~\cite{AH23} showed that quantum advantage experiments can be repurposed into certified randomness protocols under strong computational assumptions.
Bassirian et al.~\cite{BBFGT26} produced certifiable randomness in the QROM model, but with exponential-time verification. They also gave supporting evidence for the hardness assumptions of Aaronson and Hung~\cite{AH23}.

Both paradigms have limitations: device-independent protocols place demanding spatial restrictions on the setup, while computational protocols inherit the strength (or weakness) of their underlying assumptions. In addition, the \cite{BCM+21} protocol is privately verifiable and requires rounds of back-and-forth interaction, whereas the~\cite{AH23} protocol requires the verifier to run in superpolynomial time. 
The \emph{quantum random oracle model} (QROM)~\cite{BDF+11} offers a third setting that sidesteps both issues: there is a candidate protocol that is non-interactive and efficiently publicly verifiable, requiring only a single quantum device. 


\paragraph{The candidate protocol.}
In a recent breakthrough, Yamakawa and Zhandry~\cite{YZ24} constructed a proof of quantumness in the QROM. Their protocol defines a list-recoverable code $C \subset \Sigma^n$ and asks the prover to find a codeword $\mathbf{x} \in C$ such that $H_i(x_i) = 0$ for all coordinates $i \in [n]$, where $H = (H_1, \ldots, H_n)$ is a random oracle mapping symbols to bits. They gave a quantum algorithm that solves this problem using just a single round of parallel queries, while proving that no classical algorithm can succeed even with sub-exponentially many queries.

Beyond being a proof of quantumness, the Yamakawa-Zhandry protocol has a remarkable additional feature: the honest quantum algorithm's output is inherently random. Because the algorithm produces a uniformly random sample from the (large) set of valid codewords, no particular output is predictable in advance. Yamakawa and Zhandry conjectured that this is not an artifact of their algorithm but a \emph{necessary} feature of any successful strategy: any prover that causes the verifier to accept must be sampling from a high-min-entropy distribution. If true, this would yield a certifiable randomness protocol in the QROM, where the verifier's acceptance guarantees unpredictability of the output.

\paragraph{Certifiable randomness of the Yamakawa-Zhandry protocol.}
As partial evidence, Yamakawa and Zhandry proved the certifiable randomness guarantee under the Aaronson-Ambainis (AA) conjecture~\cite{AA14}. 
The AA conjecture formalizes the idea that quantum speedups require structure: formally, it asserts that any low-degree bounded polynomial (such as a quantum algorithm's acceptance probability) contains an highly influential coordinate. This in turn implies that classical query algorithms can simulate quantum query algorithms on average (i.e., on most inputs), establishing that exponential quantum advantages cannot arise in unstructured settings, but must instead rely on an underlying global structure in the input domain.
This conjecture is widely considered plausible, and
partial results support it in restricted settings (such as for block-multilinear forms with equal-magnitude coefficients~\cite{Mon12}, completely bounded block-multilinear forms~\cite{BSW22}, and random restrictions of low-degree polynomials~\cite{Bhattacharya25})
but it remains unproven in
general. The \cite{YZ24} reduction to the AA conjecture was thus a meaningful
step, but it left open the question of whether the certifiable randomness
guarantee could be established from first principles, without appeal to unproven
structural assumptions about quantum query complexity. 
 
The reliance on the AA conjecture is qualitatively different from relying on standard computational assumptions like factoring or LWE. Those are conjectures about the \emph{hardness of specific problems}; the AA conjecture is a sweeping claim about the \emph{structure of all quantum speedups}. Despite significant attention since its introduction over a decade ago~\cite{AA14}, it remains unproven in full generality, with partial results confirmed only in restricted settings. If the conjecture turned out to be false, the randomness guarantee of the YZ protocol would collapse: not because someone found a clever attack, but because the landscape of quantum query complexity was structured differently than expected.
 
It is therefore natural to ask:
 
\begin{center}
\emph{Can the certifiable randomness guarantee of the Yamakawa-Zhandry protocol\\ be established unconditionally, without the AA conjecture?}
\end{center}

\paragraph{Our result.}
We make progress on this question by proving the certifiable randomness guarantee \emph{unconditionally}---without the AA conjecture or any other unproven conjecture---against a restricted but meaningful class of quantum adversaries: those that make at most $o(\log \lambda)$ \emph{adaptive} query rounds to the random oracle. Crucially, the adversary may make polynomially many parallel quantum queries within each round; only the number of sequential rounds is bounded. Our security proof holds against computationally unbounded adversaries subject only to this query depth constraint.
 
\begin{theorem}[Informal]
The Yamakawa-Zhandry protocol satisfies $(D, h_\infty)$-certifiable min-entropy for any query depth $D(\lambda) = o(\log \lambda)$ and any min-entropy bound $h_\infty(\lambda) = o(\lambda^{c/2})$, where $c$ is a constant determined by the code parameters. That is, any adversary making $o(\log \lambda)$ rounds of (arbitrarily wide) parallel queries that causes the verifier to accept with noticeable probability must be sampling its output from a distribution with $\Omega(\lambda^{c/2})$ bits of min-entropy.
\end{theorem}

Note that the honest Yamakawa-Zhandry algorithm uses just a single round of parallel queries. Our result therefore covers all adversaries with sub-logarithmically-many adaptive rounds beyond what the honest algorithm requires.


\bigskip
\noindent\textbf{Organization.} The remainder of this paper is organized as
follows. Section~\ref{sec:overview} gives a technical overview of our proof
strategy. Section~\ref{sec:prelims} introduces the formal model, code
parameters, and query-weight definitions used throughout. Section~\ref{sec:main}
contains the proof of our main result ruling out $o(\log \lambda)$-query
strategies. 

\section{Technical Overview}\label{sec:overview}

\paragraph{The Yamakawa-Zhandry Problem.}

Fix a code $C \subset \Sigma^n$ over a large alphabet $\Sigma$, built from an appropriate list-recoverable code over $\mathbb{F}_q$. A random oracle $H = (H_1, \ldots, H_n)$ independently maps each symbol $x_i \in \Sigma$ to a bit $H_i(x_i) \in \{0,1\}$. The \emph{Yamakawa--Zhandry problem}~\cite{YZ24} is to find a codeword $\mathbf{x} \in C$ such that every symbol hashes to zero:
\[
\mathbf{x} \in C \quad \text{and} \quad H_i(x_i) = 0 \;\;\text{for all } i \in [n].
\]

\paragraph{The quantum algorithm.}
Yamakawa and Zhandry~\cite{YZ24} showed a quantum algorithm that can solve this problem efficiently, and their algorithm can be implemented with a \emph{single layer of parallel queries}. In particular, parallel queries suffice to create the following quantum state
\[
\sum_{x_1: H_1(x_1) = 0}\ket{x_1} \otimes \sum_{x_2: H_2(x_2) = 0}\ket{x_2} \otimes \ldots \otimes \sum_{x_n: H_n(x_n) = 0}\ket{x_n}
\]
which can then be post-processed, without any additional queries, to obtain the state
\[\sum_{\substack{x \in C\\
\forall i, H_i(x_i) = 0}}\ket{x}\]
Measuring this state collapses the superposition onto a random codeword satisfying the desired constraints.

Yamakawa-Zhandry also prove that no classical algorithm can solve this search problem, even with sub-exponentially many queries to the random oracle. Their proof crucially uses the {\bf list-recoverability} of the code $C$.

 
To explain what list-recoverability means, consider the following question: suppose that for each coordinate $i \in [n]$, someone hands you a ``short list'' $S_i \subseteq \Sigma$ of candidate symbols. How many codewords $\mathbf{x} \in C$ can be ``mostly consistent'' with these lists---meaning $x_i \in S_i$ for at least $(1-\zeta)n$ coordinates? List recoverability says that very few can. Concretely, if each list has size $|S_i| \leq \ell = 2^{\left(\lambda^c\right)}$, then at most $L = 2^{\tilde{O}(\lambda^{c'})}$ codewords are mostly consistent, where $c < c' < 1$ and $\zeta = \Omega(1)$ are constants determined by the code parameters.



For a classical adversary making $Q$ queries, each query reveals whether one symbol hashes to zero, so after $Q$ queries the adversary has discovered at most $Q$ symbols mapping to $0$. List recoverability guarantees that at most $L$ codewords have most of their symbols among these $Q$ discovered zeros. For each such codeword, the remaining unchecked positions must all independently hash to zero, which happens with exponentially small probability. So a classical adversary is overwhelmingly unlikely to find a correct codeword.

\paragraph{Certified Randomness.}
A crucial feature of Yamakawa-Zhandry's algorithm is that its output is \emph{inherently random}: because the measurement produces a random sample from a large set of valid codewords, no particular codeword is predictable in advance. The honest algorithm gives only negligible query weight to whatever codeword it eventually outputs---it does not ``aim'' at any specific answer.

Yamakawa and Zhandry conjectured that the randomness of the honest algorithm is not an artifact of their particular construction, but a necessary feature of \emph{any} successful quantum strategy: any adversary that finds a valid codeword must be sampling from a high-min-entropy distribution. If true, this transforms the protocol into a source of \emph{certifiable randomness}---the classical verifier checks $\mathbf{x} \in C$ and $H(\mathbf{x}) = 0$, and acceptance alone guarantees that the output was unpredictable. They proved this conjecture under the Aaronson--Ambainis (AA) conjecture~\cite{AA14}, a broad but unproven structural hypothesis about quantum speedups.

Our goal is to obviate the need for the AA conjecture. We aim to prove the certifiable randomness property of the YZ algorithm assuming only that the adversary is bounded in the number of queries they make to the random oracle. We prove that if the adversary is depth-limited, making $o(\log \secp)$ layers of parallel queries, then they cannot break the certifiable randomness property.


\paragraph{Our Techniques, in a Nutshell.}
At a high level, we show that any adversary trying to cheat---to output a low-entropy/\emph{predictable} codeword---must interact with the oracle in a way that is fundamentally different from the honest algorithm. While the honest algorithm spreads its query weight thinly and uniformly, a cheating adversary must concentrate query weight on its intended output. We show that this concentration is both {\em necessary} and {\em impossible} with just a few adaptive queries, leading to our desired contradiction.
\begin{itemize}
\item \noindent\textbf{The Heavy Query Lemma} (Lemma~\ref{thm:must-heavy-query}) shows that any adversary that breaks the certifiable min-entropy guarantee---by outputting a correct answer $\mathbf{x}$ with non-negligible conditional probability---must have allocated non-negligible \emph{query weight} to all but a small (say, $\frac{1}{100}$) fraction of the symbols of $\mathbf{x}$. In other words, a low-entropy adversary must ``check'' that almost all symbols of $\mathbf{x}$ hash to zero before committing to it. Otherwise, they would continue to output this answer with high probability even if the oracle is changed to hash some symbols of $\bfx$ to $1$.
 \item 
\noindent\textbf{The Hardness of Heavy Querying.} (Lemma~\ref{thm:2-query-security}) shows that an adversary making $o(\log \lambda)$ query layers \emph{cannot} give such heavy query weight to the symbols of any correct codeword. Intuitively, with such few adaptive rounds, the adversary must allocate its query weight before learning enough about the oracle to concentrate on a correct codeword.
\end{itemize}

We now describe our techniques in detail.

\paragraph{Query Weight and the Swapping Lemma.}
 
The central analytic quantity in our proof is \emph{query weight}. For a quantum adversary $\mathcal{A}^H$ and a symbol $(i, x_i)$, the cumulative query weight $w^{\leq j}_{i,x_i}(H)$ measures the cumulative probability, across the first $j$ query layers, that measuring a query register would yield $(i, x_i)$. Query weight is a scarce resource: an adversary making $Q$ total queries can distribute at most $Q$ total query weight across all input symbols.
 
The \emph{swapping lemma} (Lemma~\ref{thm:swapping-lemma}) bounds how much the adversary's quantum state can change when the oracle is reprogrammed on a set $X$ of inputs, as a function of the total query weight on those inputs. We have that:
\begin{equation}
\label{eq:fir}
\left\| |\psi^H_j\rangle - |\psi^{H'}_j\rangle \right\|_2 \;\leq\; 2 \cdot \sqrt{j \cdot \sum_{(i,x) \in X} w^{\leq j}_{i,x}(H)}.
\end{equation}
where $\ket{\psi^\cO_j}$ denotes the adversary's state after $j$ query layers when it interacts with oracle $\cO$.
If the adversary gives small query weight to the reprogrammed positions, its state---and therefore its output distribution---barely changes. This allows us to {\em reprogram} the oracle, as we describe in more detail below.
 
\paragraph{Step 1: Low-Entropy Adversaries Must Heavily Query Their Answer.}
Suppose an adversary outputs a correct codeword $\mathbf{x}$ with high conditional probability $p$ given an oracle $h$, i.e., \[\Pr[\mathcal{A}^H \to \mathbf{x} \mid H = h] \geq p,\] but gives small query weight during $D$ query layers \[w^{\leq D}_{i,x_i}(h) < \frac{p^2}{16Qn}\] to at least $s+1$ symbols of $\mathbf{x}$. We show this is impossible by deriving a contradiction via a counting argument.
 
\paragraph{Constructing bad oracles.}
Let $I_2 \subseteq [n]$ be the set of $(s+1)$ symbol positions where the adversary gives small query weight. For each non-empty subset of $\{(i, x_i)\}_{i \in I_2}$, reprogram $h$ to flip the output from $0$ to $1$ on those inputs. Call the resulting oracle $h'$. This produces a set $T_{h,\mathbf{x}}$ of at least $2^s$ ``bad'' oracles, each of which makes $\mathbf{x}$ incorrect (since $h'(\mathbf{x}) \neq 0$).
 
\paragraph{The adversary doesn't notice.}
By the swapping lemma, since the adversary gives small total query weight to the reprogrammed positions, its output distribution changes by at most \[2\sqrt{D \cdot |I_2| \cdot (p^2/16Qn)} \leq 2\sqrt{Q \cdot n \cdot (p^2/16Qn)} \leq p/2.\]
Therefore $\Pr[\mathcal{A}^{h'} \to \mathbf{x}] \geq p/2$ on every bad oracle.
 
\paragraph{Counting.}
 
Let $S$ be the set of all ``good'' pairs $(h, \mathbf{x})$ that satisfy all three conditions simultaneously: (1)~the adversary gives small query weight to many positions of $\mathbf{x}$, (2)~$\Pr[\mathcal{A}^h \to \mathbf{x}] \geq p$, and (3)~$\mathbf{x}$ is a correct codeword under $h$. Our goal is to show that $\Pr[(H, X) \in S]$, which is the probability that a randomly sampled oracle-output pair lands in $S$, is negligible.
 
We have shown that each good pair $(h, \mathbf{x}) \in S$ generates a set $T_{h,\mathbf{x}}$ of at least $2^s$ bad oracles, and that $\Pr[\mathcal{A}^{h'} \to \mathbf{x}] \geq \frac{1}{2}\Pr[\mathcal{A}^{h} \to \mathbf{x}]$ on each of them. The bad sets are disjoint: given any $(h', \mathbf{x}) \in T_{h,\mathbf{x}}$, the original oracle $h$ is uniquely determined by reprogramming (back) $h'$ to map $\mathbf{x}$ back to $0$. Since different good pairs produce disjoint bad sets, we can sum over all of them without overcounting. Each term $\frac{1}{|\mathcal{H}|} \cdot \Pr[\mathcal{A}^{h'} \to \mathbf{x}]$ is the joint probability of drawing oracle $h'$ and the adversary outputting $\mathbf{x}$, and disjoint events sum to at most $1$:
\[
1 \;\geq\; \sum_{(h,\mathbf{x}) \in S} \;\sum_{h' \in T_{h,\mathbf{x}}} \frac{1}{|\mathcal{H}|} \cdot \Pr[\mathcal{A}^{h'} \to \mathbf{x}] \;\geq\; \sum_{(h,\mathbf{x}) \in S} 2^s \cdot \frac{1}{2|\mathcal{H}|} \cdot \Pr[\mathcal{A}^{h} \to \mathbf{x}] \;=\; 2^{s-1} \cdot \Pr[(H,X) \in S].
\]
Rearranging gives $\Pr[(H,X) \in S] \leq 2^{-(s-1)}$, which is negligible since $s = \omega(\log \lambda)$.

Thus, with overwhelming probability, any adversary that outputs a high-probability correct codeword must give non-negligible query weight to all but $s$ of its symbols. We call this behavior \emph{heavily querying} the answer.

\paragraph{Step 2: Low-Depth Adversaries Cannot Heavily Query a Correct Answer.}
 
This is the heart of the argument, where we exploit the adversary's limited query depth.
 
\paragraph{The key intuition.}
As discussed above, the honest Yamakawa--Zhandry algorithm gives only negligible query weight to the codeword it eventually outputs. A low-entropy adversary, by contrast, must concentrate non-negligible weight on a small set of codewords (as already shown in Step 1). However, with few query layers, the adversary is forced to make these ``bets'' early, before learning where the oracle maps symbols to zero. Such an adversary is essentially gambling.

We first describe a warm-up that captures the core idea in the simplest possible setting.
 
\paragraph{Warm-up: a single parallel query.}
Consider an adversary that makes just one layer of parallel queries. It prepares some initial quantum state $|\psi_0\rangle$, queries the oracle once on all coordinates in superposition, applies some unitary, and measures to produce its output. The key observation is that the adversary's query weight is \emph{oracle-independent}: the state $|\psi_0\rangle$ is fixed before the oracle is queried, so the query weight $w^1_{i,x_i}$ given to each symbol $(i, x_i)$ is the same regardless of which oracle $h$ was chosen.
 
 
Now consider reprogramming the oracle on a subset of $\zeta n$ symbols of $\bfx$. Flipping even one of these positions from $0$ to $1$ makes $\mathbf{x}$ incorrect. We can therefore construct $2^{\Omega(n)}$ bad oracles from each good oracle, and the adversary still heavily queries $\mathbf{x}$ on all of them. By list recoverability, each bad oracle is associated with at most $L = 2^{\tilde{O}(\lambda^{c'})}$ good oracles (since at most $L$ codewords are heavily queried). A similar counting argument as in Step~1 gives:
\[
\frac{|S|}{|\mathcal{H}|} \;\leq\; \frac{L}{2^{\Omega(n)}} \;=\; \frac{2^{\tilde{O}(\lambda^{c'})}}{2^{\Omega(\lambda)}} \;=\; \mathrm{negl}(\lambda).
\]
This completes the single-query case. The argument is clean because query weights are fixed upfront, so reprogramming the oracle on low-weight positions is ``free.''

\paragraph{The challenge with multiple queries.}
With $D > 1$ adaptive query layers, the argument above breaks down. After the first query, the adversary's state---and therefore the query weights in subsequent layers---can depend on the oracle. Reprogramming the oracle on a set of positions may change not only the adversary's final output but also how it allocates query weight in later rounds. We can no longer argue that the adversary heavily queries $\mathbf{x}$ independently of the oracle.
 
The solution is a \emph{bootstrapping} argument. Instead of arguing about query weights all at once, we track how they build up layer by layer, using a sequence of growing thresholds to control the perturbation at each step.

\paragraph{The idea: find the right moment to reprogram.}
The single-query-layer warm-up suggests the following strategy: find positions with low query weight, reprogram them, and argue that the adversary's behavior is unchanged by the reprogramming. With multiple adaptive layers, query weights depend on the oracle, so we cannot simply point to low-weight positions at the end of the computation---they may have been low only because the oracle steered the adversary elsewhere.
 
Instead, we look for the right \emph{moment} during the computation. The adversary starts with zero query weight on $\mathbf{x}$ and ends with high query weight on most of $\mathbf{x}$'s symbols (by Step~1). Somewhere in between, there is a \emph{critical layer} $q^*$: the first layer after which the adversary has accumulated substantial weight on $(1-\zeta)n$ symbols of $\mathbf{x}$. Just before this layer, after $q^*-1$ layers, the adversary has \emph{not yet} reached this threshold, so at least $\zeta n$ symbols of $\mathbf{x}$ still have low accumulated weight. These are the positions we reprogram.
 
\paragraph{Controlling the perturbation with thresholds.}
The crux of the argument is showing that reprogramming these low-weight positions does not significantly affect the adversary's query weights on the \emph{other} positions---the ones that \emph{are} heavily queried. If we can show this, then the adversary still heavily queries $\mathbf{x}$ under the reprogrammed oracle $h'$, even though $\mathbf{x}$ is no longer correct under $h'$, and the counting argument goes through as in the warm-up.
 
The swapping lemma tells us that on each individual query layer $q \leq q^*$, the perturbation to any symbol's query weight is at most $2W \sqrt{D n \cdot t_{q^*-1}}$, where $t_{q^*-1}$ is an upper bound on the weight that each reprogrammed position had accumulated through layer $q^*-1$, $W$ is the number of parallel queries per layer (the width) and $D$ is the total number of query layers (the depth). Summing over $q^* \leq D$ layers, the total perturbation to the cumulative weight of any heavily-queried symbol is at most $2D W \sqrt{Dn \cdot t_{q^*-1}} = 2Q\sqrt{Dn \cdot t_{q^*-1}}$.
 
For this perturbation to be harmless, we need the heavily-queried symbols to have accumulated enough weight by layer $q^*$ to ``absorb'' it---that is, their weight at layer $q^*$ must exceed $Q/\ell + 2Q\sqrt{Dn \cdot t_{q^*-1}}$, so that even after subtracting the perturbation, they still have weight $\geq Q/\ell$.
 
\paragraph{The threshold sequence.}
To make this work, we define a sequence of growing thresholds $t_0 < t_1 < \cdots < t_D$ and say that the adversary has ``reached threshold $q$'' if at least $(1-\zeta)n$ symbols of $\mathbf{x}$ have cumulative weight $\geq t_q$ after $q$ layers. The critical layer $q^*$ is the first layer where the adversary reaches its threshold.
 
The thresholds are defined by
\[
t_q = \left(\frac{2Qn}{\ell}\right)^{4^{-q}}
\]
which is a doubly-exponential interpolation between $t_0 = 2Qn/\ell$ (very small, since $\ell$ is superpolynomial) and $t_Q$, which we verify is at most $p^2/16Qn$ (the threshold from Step~1). The key property is the recurrence, for all $q \in [2, Q]:$
\[
t_q \;\geq\; \frac{Q}{\ell} \;+\; 2Q\sqrt{Dn \cdot t_{q-1}},
\]
which is exactly the condition we need: if the adversary reaches threshold $q$ at layer $q^*$, and the reprogrammed positions have weight below $t_{q^*-1}$, then the perturbation $2Q\sqrt{Dn \cdot t_{q^*-1}}$ is small enough that the heavily-queried symbols retain weight $\geq Q/\ell$ after reprogramming. This recurrence is what limits us to $Q = o(\log \lambda)$ query layers---with more layers, the thresholds grow too quickly and the argument breaks down.
 
\paragraph{Completing the counting.}
Once we know the adversary still heavily queries $\mathbf{x}$ under the bad oracle $h'$, we need to bound how many good pairs $(h, \mathbf{x})$ can map to the same bad oracle $h'$. This is where list recoverability enters.
 
Fix any oracle $h'$ (good or bad). The adversary $\mathcal{A}^{h'}$ makes $D$ total layers of $W$ parallel queries each, so it distributes at most $Q = D \cdot W$ total query weight across all symbols. For each coordinate $i \in [n]$, let $S_i$ be the set of symbols $x \in \Sigma$ that receive cumulative weight $w^{\leq D}_{i,x}(h') \geq Q/\ell$. Since the total weight is $\leq Q$, each coordinate can have at most $\ell$ such symbols: $|S_i| \leq \ell$. Now, if $\mathcal{A}^{h'}$ heavily queries a codeword $\mathbf{x}$, then $x_i \in S_i$ for at least $(1-\zeta)n$ coordinates---the codeword is ``mostly consistent'' with the lists $S_1, \ldots, S_n$. By list recoverability, at most $L = 2^{\tilde{O}(\lambda^{c'})}$ codewords can be mostly consistent with these lists.
 
Each bad oracle $h'$ can belong to $T_{h,\mathbf{x}}$ for at most $L$ distinct good pairs $(h, \mathbf{x})$, because each such pair is associated with a distinct heavily-queried codeword $\bfx$, and there are at most $L$ of those. Since each good pair generates at least $2^{\lfloor \zeta n \rfloor}$ bad oracles, we get:
\[
\frac{|S|}{|\mathcal{H}|} \;\leq\; \frac{L}{2^{\lfloor \zeta n \rfloor}} \;=\; \frac{2^{\tilde{O}(\lambda^{c'})}}{2^{\Omega(\lambda)}} \;=\; \mathrm{negl}(\lambda)
\]
since $c' < 1$. Thus, with overwhelming probability over the choice of oracle, no correct codeword is heavily queried by the adversary.
 
\paragraph{Putting the Pieces Together.}
 
Suppose toward contradiction that some $o(\log \lambda)$-depth adversary breaks the certifiable min-entropy guarantee with non-negligible probability. Then with non-negligible probability over the oracle $h$, the adversary succeeds with noticeable probability and its output has low min-entropy. Step~1 says that such an adversary must heavily query its answer. But Step~2 says that with overwhelming probability over $h$, no correct codeword is heavily queried---a contradiction. This completes the proof.

\paragraph{Discussion.}
Our techniques are limited to $o(\log \lambda)$ adaptive rounds because our threshold bootstrapping technique requires thresholds that grow doubly exponentially with the number of rounds. When the number of rounds exceeds $O(\log \lambda)$, the thresholds exceed the bounds needed for the counting argument. Extending the certifiable randomness guarantee to adversaries making $\omega(\log \lambda)$ or polynomially many adaptive rounds remains an important open problem; following Yamakawa and Zhandry, the guarantee for such adversaries currently still relies on the AA conjecture.

\section{Preliminaries}\label{sec:prelims}
\subsection{Certifiable Min-Entropy}\label{sec:prelim-certifiable-min-entropy}
A certifiable min-entropy protocol is an interactive proof system where any prover that causes the verifier to accept with noticeable probability must be sampling their accepting proofs from a distribution with high min-entropy. In the random oracle model, the prover and verifier have quantum query access to the random oracle, and the min-entropy of the proof is computed after conditioning on the choice of oracle. That way, the proof will provide additional randomness beyond the randomness of the oracle. 

\Cref{def:certifiable-min-entropy} below recalls the definition of a certifiable min-entropy protocol from \cite{YZ24} (almost verbatim). We modify the min-entropy property to allow us to specify the prover's query depth $D$ and min-entropy bound $\minent$.

\begin{definition}[Certifiable Min-Entropy Protocol, \cite{YZ24} Definition 3.5]\label{def:certifiable-min-entropy}
    A (keyless, non-interactive, publicly verifiable) \textbf{certifiable min-entropy protocol} relative to a random oracle consists of the algorithms $(\Prove, \Verify)$, which have the following syntax.

    \paragraph{Syntax.}
    \begin{itemize}
        \item $\Prove^\RO (1^\secp, 1^\minent) \to \pi$: This is a QPT algorithm that takes the security parameter $1^\secp$ and a min-entropy threshold $1^\minent$ as input. It makes $\poly(\secp, \minent)$ quantum queries to the random oracle $\RO$, and outputs a classical proof $\pi$.
        \item $\Verify^\RO (1^\secp, 1^\minent, \pi) \to x$: This is a deterministic classical polynomial-time algorithm that takes $1^\secp$, $1^\minent$, and a proof $\pi$ as input; it makes $\poly(\secp, \minent)$ queries to the random oracle $\RO$, and outputs either a string $x$ (whose length may depend on $\minent$), or $\bot$ indicating rejection.
    \end{itemize}
    
    We require a certifiable min-entropy protocol to satisfy the following properties:

    \paragraph{Correctness.} For any $\minent = \minent(\secp)$, we have
    \[\Pr_{\RO}\left[\Verify^\RO (1^\secp, 1^\minent, \pi) = \bot : \pi \gets \Prove^\RO (1^\secp, 1^\minent)\right] \leq \negl(\secp).\]

    \paragraph{$(D, \minent)$-Certifiable Min-Entropy.} Given a polynomially-bounded query depth $D = D(\secp)$ and a polynomially-bounded min-entropy $\minent = \minent(\secp)$, the protocol satisfies \textbf{$(D, \minent)$-certifiable min-entropy} if for any polynomially-bounded query width $W = W(\secp)$, any unbounded-time adversary $\cA$ that makes $D$ layers of $W$ parallel quantum queries to $\RO$, and any inverse polynomial function $\delta(\cdot)$, there is a negligible function $\negl(\cdot)$ such that the following holds. Let $\cA^\RO_\top (1^\secp, 1^\minent)$ be the distribution $\Verify^\RO[1^\secp, 1^\minent, \cA^\RO(1^\secp, 1^\minent)]$, conditioned on the output not being $\bot$. Then:
    \begin{align*}
        \Pr_{\RO}\left[\Pr\left[\Verify^\RO [1^\secp, 1^\minent, \cA^\RO(1^\secp, 1^\minent)] \neq \bot\right] \geq \delta(\secp) \land \Minent\left(\cA^\RO_\top (1^\secp, 1^\minent)\right) \leq \minent(\secp)\right] \leq \negl(\secp).
    \end{align*}
    If $(D, \minent)$ are not specified, we say that the protocol satisfies \textbf{certifiable min-entropy} if for any polynomial functions $D, \minent$, the protocol satisfies $(D, \minent)$-certifiable min-entropy. 
\end{definition}

\subsection{The Yamakawa-Zhandry Protocol}\label{sec:prelim-YZ-protocol}
Here we present a candidate construction of a certifiable min-entropy protocol that is based on \cite{YZ24}'s proof of quantumness. The protocol defines a list-recoverable code $C$ and a random oracle $\RO$ and asks the adversary to find a codeword $\bfx \in C$ such that $\RO(\bfx) = \mathbf{0}$.

\paragraph{Parameters:} Let $\secp \in \bbN$ be the security parameter. Let $q = 2^{2 \lfloor \log \secp \rfloor}$, $N = q - 1$, $m = 2^{\lfloor \log \secp \rfloor} + 1$, $n = \frac{N}{m} = 2^{\lfloor \log \secp \rfloor} - 1$. Let $k = \alpha \cdot N$, where $\alpha$ is an arbitrary constant in the range $\left(\frac{5}{6}, 1\right)$ such that $\alpha \cdot N \in \bbN$. Let $\gamma$ be an arbitrary generator of $\bbF_q^*$.
    
\paragraph{Folded Reed-Solomon Code:} Let $\Sigma = \bbF_q^m$ be the alphabet of the code. Let the code be $C \subset \Sigma^n$.
    
We can interpret a word $\bfx \in \Sigma^n$ as a vector $\tilde{\bfx} \in \bbF_q^{m \cdot n} = \bbF_q^N$. We say that $\bfx \in C$ if and only if there is a polynomial $f \in \bbF_q[X]$ of degree $\leq k$ such that 
    \[\tilde{\bfx} = \left(f(\gamma), f(\gamma^2), \dots, f(\gamma^{N})\right)\]

Note that any two distinct codewords differ on at least $(1-\alpha) \cdot n$ symbols.

\begin{theorem}[List Recoverability, \cite{YZ24}, lemma 4.2]\label{thm:list-recoverability}
    Let us define the specific list-recoverability parameters of the code. Let $c, c', \zeta, \ell, L$ be the parameters of the list recovery property of the code (in \cite{YZ24}, lemma 4.2). We have that $c, c'$ are constants satisfying $0 < c < c' < 1$. Next, $\zeta = \Omega(1), \ell = 2^{\left(\secp^c\right)}, L = 2^{\tilde{O}(\secp^{c'})}$. 

    Then for any sets $S_1, \dots, S_n \subset \Sigma$ of size $\leq \ell$, the number of codewords $\bfx \in C$ for which $x_i \in S_i$ for at least $(1-\zeta)n$ values of $i \in [n]$ is $\leq L$.
\end{theorem}

\paragraph{Random Oracle:} For each $i \in [n]$, let $\cH^{(i)}$ be the set of all functions $\RO_i$ that map $\Sigma \to \bit$. Let $\cH = \cH^{(1)} \times \dots \times \cH^{(n)}$. The random oracle $\RO = (\RO_1, \dots, \RO_n)$ is sampled uniformly at random from $\cH$. Given an input $\bfx = (x_1, \dots, x_n) \in \Sigma^n$, let $\RO(\bfx) = \left(\RO_1(x_1), \dots, \RO_n(x_n)\right)$.

\paragraph{Candidate Construction:}\label{cert-min-ent-construction}
Here is the candidate construction of a certifiable min-entropy protocol. This is essentially the same as \cite{YZ24}'s proof of quantumness protocol. 

\begin{itemize}
    \item $\mathsf{Setup}$: Given the security parameter $\secp$, compute the parameters $q, N, m, n, \alpha, k$ as described above. This determines the code $C$ and the oracle's sample space $\cH$. Then sample $\RO \getsr \cH$.
    \item $\Prove^\RO (1^\secp, 1^\minent)$: This algorithm is the same as the $\Prove^\RO(1^\secp)$ algorithm constructed in \cite{YZ24} section 6, except that our algorithm makes its queries to $\RO$ in parallel, and our algorithm looks for preimages of $0$ rather than $1$.
    \item $\Verify^\RO (1^\secp, 1^\minent, \pi)$: This algorithm is the same as the $\Verify^\RO (1^\secp, \pi)$ algorithm constructed in \cite{YZ24} section 6, except that our algorithm checks that $\RO_i(x_i) = 0$ (rather than $1$) for all $i \in [n]$.
    \begin{enumerate}
        \item Parse $\pi = \bfx = (x_1, \dots, x_n)$.
        \item Check that $\bfx \in C$ and $\RO_i(x_i) = 0$ for all $i \in [n]$.
        \item If all checks pass, then output $\bfx$. Otherwise, output $\bot$.
    \end{enumerate}
\end{itemize}

\cite{YZ24} constructed their certifiable min-entropy protocol by using complexity leveraging to select the security parameter of the underlying proof of quantumness protocol. We omit the complexity leveraging here in order to study the min-entropy properties of the basic protocol.

\begin{lemma}\label{thm:syntax-and-correctness-of-CR-protocol}
    The candidate construction described above satisfies the syntax and correctness properties of a certifiable min-entropy protocol (\cref{def:certifiable-min-entropy}).
\end{lemma}
\begin{proof}
    It is clear by inspection that the syntax property is satisfied. Next, \cite{YZ24} lemma 6.3 implies that the correctness property is satisfied as well.
\end{proof}

\subsection{Query-Bounded Algorithms}\label{sec:prelim-query-bounded-algorithms}
Here we define a quantum algorithm $\cA$ that makes $D$ layers of $W$ parallel queries.

Let $\cA^\RO$ denote a quantum algorithm $\cA$ with quantum query access to oracle $\RO$. $\cA$ is parametrized by a query width $W \geq 1$, which is the number of parallel queries $\cA$ can make, and a query depth $D \geq 1$, which is the number of query layers it can make in sequence. Let $Q = W \cdot D$ be the total number of queries $\cA$ can make. 

Without loss of generality, let $\cA^\RO$ operate as follows: 
\begin{itemize}
    \item $\cA$ starts with an initial pure state $\ket{\psi_0} = \ket{\psi^\RO_0}$ on registers $R = R_{Q_1} \times \ldots \times R_{Q_W} \times R_A$, a sequence of unitaries $(U_1, \dots, U_D)$, and a measurement operator $M$.
    \item For each layer $j \in [D]$:
    \begin{itemize}
        \item $\cA$ submits the query registers $R_{Q_1} \times \ldots \times R_{Q_W}$ of $\ket{\psi_{j-1}^\RO}$ to the oracle $\RO$. Then $\RO$ acts as a phase oracle on each $R_{Q_k}$, for all $k \in [W]$, and returns $R_{Q_1} \times \ldots \times R_{Q_W}$ to $\cA$.
        \item $\cA$ applies $U_j$ to its state to obtain a new state $\ket{\psi_{j}^\RO}$.
    \end{itemize}
    \item Finally, $\cA$ applies measurement $M$ to $\ket{\psi_{D}^\RO}$ and outputs the measurement outcome.
\end{itemize}

Next, we define the query weight $w_{i,x}^{j}(\RO)$ to upper-bound the probability that we obtain query $(i,x)$ on at least one query register if we measure $\cA^\RO$'s $j$-th query layer.
\begin{definition}[Query Weight]
For each possible classical query $(i,x) \in [n] \times \Sigma$ and each query layer $j \in [D]$, let $w_{i,x}^{j}(\RO)$ be the \textbf{query weight} that $\cA^\RO$ gives to $(i,x)$ on the $j$-th query layer.
\[w_{i,x}^{j}(\RO) = \sum_{k \in [W]} \Tr\left[\left(\ketbra{i,x}_{R_{Q_k}} \otimes \mathbb{I}_{R \backslash R_{Q_k}}\right) \ketbra{\psi_{j-1}^\RO}\right]\]
Also, let $w_{i, x}^0(\RO) = 0$ for all $(i, x, \RO)$. Finally, let $w_{i,x}^{\leq j}(\RO)$ be the \textbf{cumulative query weight} that $\cA^\RO$ gives to $(i,x)$ on the first $j$ query layers.
\[w_{i,x}^{\leq j}(\RO) = \sum_{j' = 0}^j w_{i,x}^{j'}(\RO)\]
\end{definition}

The following lemma says that the query weight must be $\geq 0$ because the query weight is the sum of probabilities.
\begin{lemma}
    For any $\cA, i,x, \RO, j$,
    \[0 \leq w_{i,x}^{j}(\RO)\]
\end{lemma}
\begin{proof}
The projector $\ketbra{i,x}_{R_{Q_k}} \otimes \mathbb{I}_{R \backslash R_{Q_k}}$ is positive-semidefinite, so for any state $\ket{\psi}$,
\[\bra{\psi} \left(\ketbra{i,x}_{R_{Q_k}} \otimes \mathbb{I}_{R \backslash R_{Q_k}}\right) \ket{\psi} \geq 0\]

Next,
\begin{align*}
    w_{i,x}^{j}(\RO) &= \sum_{k \in [W]} \Tr\left[\left(\ketbra{i,x}_{R_{Q_k}} \otimes \mathbb{I}_{R \backslash R_{Q_k}}\right) \ketbra{\psi_{j-1}^\RO}\right]\\
    &= \sum_{k \in [W]} \bra{\psi_{j-1}^\RO} \left(\ketbra{i,x}_{R_{Q_k}} \otimes \mathbb{I}_{R \backslash R_{Q_k}}\right) \ket{\psi_{j-1}^\RO}\\
    &\geq 0
\end{align*}
\end{proof}

The following lemma says that the cumulative query weight cannot decrease as the number of queries increases. This is because the query weight on a particular query is always non-negative.
\begin{lemma}\label{thm:cumulative-query-weight-increases-over-time}
    For any $\cA, i,x, \RO$ and any two $j, j' \in \{0\} \cup [D]$ for which $j \leq j'$,
    \[w_{i,x}^{\leq j}(\RO) \leq w_{i,x}^{\leq j'}(\RO)\]
\end{lemma}
\begin{proof}
If $j = j'$, then $w_{i,x}^{\leq j}(\RO) = w_{i,x}^{\leq j'}(\RO)$, so the claim is true. From now on, let us consider the case where $j < j'$.
    \begin{align*}
        w_{i,x}^{\leq j'}(\RO) - w_{i,x}^{\leq j}(\RO) &= \sum_{j'' \in \{j+1, \dots, j'\}} w_{i,x}^{j''}(\RO)\\
        &\geq \sum_{j'' \in \{j+1, \dots, j'\}} 0 = 0\\
        w_{i,x}^{\leq j'}(\RO) &\geq w_{i,x}^{\leq j}(\RO)
    \end{align*}
\end{proof}

The following lemma says that the total query weight on all symbols during a given query is $W$.
\begin{lemma}\label{thm:total-query-weight-per-layer}
    For any $\cA, \RO$ and any $j \in [D]$,
    \[\sum_{(i,x) \in [n] \times \Sigma} w^{j}_{i,x}(\RO) = W\]
\end{lemma}
\begin{proof}
    \begin{align*}
        \sum_{(i,x) \in [n] \times \Sigma} w^{j}_{i,x}(\RO) &= \sum_{(i,x) \in [n] \times \Sigma} \sum_{k \in [W]} \Tr\left[\left(\ketbra{i,x}_{R_{Q_k}} \otimes \mathbb{I}_{R \backslash R_{Q_k}}\right) \ketbra{\psi_{j-1}^\RO}\right]\\
        &= \sum_{k \in [W]} \Tr\left[\left(\sum_{(i,x) \in [n] \times \Sigma} \ketbra{i,x}_{R_{Q_k}} \otimes \mathbb{I}_{R \backslash R_{Q_k}}\right) \ketbra{\psi_{j-1}^\RO}\right]\\
        &= \sum_{k \in [W]} \Tr\left[\left(\mathbb{I}_{R_{Q_k}} \otimes \mathbb{I}_{R \backslash R_{Q_k}}\right) \ketbra{\psi_{j-1}^\RO}\right]\\
        &= \sum_{k \in [W]}\Tr\left[\ketbra{\psi_{j-1}^\RO}\right] = \sum_{k \in [W]}1\\
        &= W
    \end{align*}
\end{proof}

Likewise, the total cumulative query weight given to all symbols after $j$ queries is $j \cdot W$.

\begin{lemma}\label{thm:total-cumulative-query-weight}
    For any $\cA, \RO$ and any $j \in [D]$,
    \[\sum_{(i,x) \in [n] \times \Sigma} w^{\leq j}_{i,x}(\RO) = j \cdot W\]
\end{lemma}
\begin{proof}
    \begin{align*}
        \sum_{(i,x) \in [n] \times \Sigma} w^{\leq j}_{i,x}(\RO) &= \sum_{(i,x) \in [n] \times \Sigma} \sum_{j' = 0}^j w_{i,x}^{j'}(\RO) = \sum_{(i,x) \in [n] \times \Sigma} \sum_{j' = 1}^j w_{i,x}^{j'}(\RO)\\
        &= \sum_{j' = 1}^j \sum_{(i,x) \in [n] \times \Sigma} w_{i,x}^{j'}(\RO)\\
        &= \sum_{j' = 1}^j W\\
        &= j \cdot W
    \end{align*}
\end{proof}

The swapping lemma (\cref{thm:swapping-lemma}) says that the adversary's states on two oracles $\RO$ and $\RO'$ will be close in Euclidean distance if the adversary gives small query weight to the positions where the oracles differ.
 \begin{lemma}[Swapping Lemma, Adapted from \cite{Vaz98} Lemma 2,
\cite{BBBV97} Theorem 3.3]\label{thm:swapping-lemma}
        Given two oracles $\RO, \RO'$, let $X \subseteq [n] \times \Sigma$ be the subset of inputs on which $\RO$ and $\RO'$ differ. Then, for any $j \in \{0, \ldots, D\}$,
        \[\left\|\ket{\psi^\RO_{j}} - \ket{\psi^{\RO'}_{j}}\right\|_2 \leq 2 \cdot \sqrt{j \cdot \sum_{(i,x) \in X} w_{i,x}^{\leq j}(\RO)}\]
    \end{lemma}
    \begin{proof}
        First, let us consider the case where $j = 0$. Then 
        \[\left\|\ket{\psi^\RO_{j}} - \ket{\psi^{\RO'}_{j}}\right\|_2 = \left\|\ket{\psi_{0}} - \ket{\psi_{0}}\right\|_2 = 0 = 2 \cdot \sqrt{j \cdot \sum_{(i,x) \in X} w_{i,x}^{\leq j}(\RO)}\]
        From now on, we will consider the case where $j \geq 1$.
        
        Second, let us prove this claim for algorithms making one query per layer ($W = 1$). Our proof is adapted from \cite[Theorem~3.3]{BBBV97}. Let $O^\RO$ and $O^{\RO'}$ be unitaries that respond to a given query using oracle $\RO$ or $\RO'$ respectively. Next, for each $j' \in [j]$, let
        \begin{align*}
            \ket{E_{j'}} &= U_{j'} \cdot O^{\RO'} \cdot \ket{\psi_{j'-1}^\RO} - U_{j'} \cdot O^\RO \cdot \ket{\psi_{j'-1}^\RO}\\
            &= U_{j'} \cdot \left(O^{\RO'} - O^\RO\right) \cdot \ket{\psi_{j'-1}^\RO}
        \end{align*}
        Note that when $W = 1$,
        \begin{align*}
            O^{\RO'} - O^\RO &= \sum_{(i,x) \in [n] \times \Sigma} \ketbra{i,x}_{R_Q} \otimes \mathbb{I}_{R_A} \cdot (-1)^{\RO'(i,x)}\\
            &\quad\quad- \sum_{(i,x) \in [n] \times \Sigma} \ketbra{i,x}_{R_Q} \otimes \mathbb{I}_{R_A} \cdot (-1)^{\RO(i,x)}\\
            &= \sum_{(i,x) \in [n] \times \Sigma} \ketbra{i,x}_{R_Q} \otimes \mathbb{I}_{R_A} \cdot (-1)^{\RO'(i,x)} \cdot \left[1 - (-1)^{\RO(i,x) - \RO'(i,x)}\right]\\
            &= \sum_{(i,x) \in X} \ketbra{i,x}_{R_Q} \otimes \mathbb{I}_{R_A} \cdot (-1)^{\RO'(i,x)} \cdot 2\\
            \left(O^{\RO'} - O^\RO\right)^\dag \cdot \left(O^{\RO'} - O^\RO\right) &= 4 \cdot \sum_{(i,x) \in X} \ketbra{i,x}_{R_Q} \otimes \mathbb{I}_{R_A}\\
        \end{align*}
        Therefore,
        \begin{align*}
            \left\|\ket{E_{j'}}\right\|_2^2 &= \left\|\left(O^{\RO'} - O^\RO\right) \cdot \ket{\psi_{j'-1}^\RO}\right\|_2^2\\
            &= \bra{\psi_{j'-1}^\RO} \cdot \left(4 \cdot \sum_{(i,x) \in X} \ketbra{i,x}_{R_Q} \otimes \mathbb{I}_{R_A}\right) \cdot \ket{\psi_{j'-1}^\RO}\\
            &= 4 \cdot \sum_{(i,x) \in X} \Tr\left[\ketbra{i,x}_{R_Q} \otimes \mathbb{I}_{R_A} \cdot \ketbra{\psi_{j'-1}^\RO}\right]\\
            &= 4 \cdot \sum_{(i,x) \in X} w_{i,x}^{j'}(\RO)
        \end{align*}
        Next,
        \begin{align*}
            \ket{\psi^\RO_{j}} &= U_{j} \cdot O^\RO \cdot \ket{\psi_{j-1}^\RO}\\
            &= U_{j} \cdot O^{\RO'} \cdot \ket{\psi_{j-1}^\RO} - \ket{E_{j}}\\
            &= \left(\prod_{j' = 1}^j U_{j'} \cdot O^{\RO'}\right) \cdot \ket{\psi_0} - \sum_{j' = 1}^j \left(\prod_{j'' = j'+1}^j U_{j''} \cdot O^{\RO'}\right) \cdot \ket{E_{j'}}\\
            &= \ket{\psi^{\RO'}_j} - \sum_{j' = 1}^j \left(\prod_{j'' = j'+1}^j U_{j''} \cdot O^{\RO'}\right) \cdot \ket{E_{j'}}\\
            \left\|\ket{\psi^\RO_{j}} - \ket{\psi^{\RO'}_j}\right\|_2 &= \left\|\sum_{j' = 1}^j\left(\prod_{j'' = j'+1}^j U_{j''} \cdot O^{\RO'}\right) \cdot \ket{E_{j'}}\right\|_2\\
            &\leq \sum_{j' = 1}^j \left\|\left(\prod_{j'' = j'+1}^j U_{j''} \cdot O^{\RO'}\right) \cdot \ket{E_{j'}}\right\|_2\\
            &\leq \sqrt{j \cdot \sum_{j' = 1}^j \left\|\left(\prod_{j'' = j'+1}^j U_{j''} \cdot O^{\RO'}\right) \cdot \ket{E_{j'}}\right\|_2^2}\\
            &= \sqrt{j \cdot \sum_{j' = 1}^j \left\|\ket{E_{j'}}\right\|_2^2}\\
            &= \sqrt{j \cdot \sum_{j' = 1}^j 4 \cdot \sum_{(i,x) \in X} w_{i,x}^{j'}(\RO)}\\
            &= 2 \cdot \sqrt{j \cdot \sum_{(i,x) \in X} w_{i,x}^{\leq j}(\RO)}
        \end{align*}
        Third, we will reduce the case where $W > 1$ to the case where $W=1$. We can view each query layer as making one query to the oracle $\RO^W$ that applies $\RO$ in parallel $W$ times.
        
        Next, the inputs $\bfz = [(i_1, x_1), \dots, (i_W, x_W)]$ on which $\RO^W$ and $(\RO')^W$ differ satisfy 
        \[(i_1, x_1) \in X \lor \ldots \lor (i_W, x_W) \in X\] 
        Let $X^{\lor W}$ be the set of all inputs satisfying the condition above.

        For any query layer $j'$, let us compute $\sum_{\bfz \in X^{\lor W}} w_{\bfz}^{j'}(\RO^W)$, the query weight given to all inputs in $X^{\lor W}$ during the $j'$-th query to $\RO^W$. This equals the probability that we obtain a value in $X^{\lor W}$ if we measure the $j'$-th query.
        \begin{align*}
            \sum_{\bfz \in X^{\lor W}} w_{\bfz}^{j'}(\RO^W) &= \Pr\left[\text{On the $j'$-th query to $\RO^W$, $R_{Q_1} \times \ldots \times R_{Q_W}$ contains a value in $X^{\lor W}$}\right]\\
            &\leq \sum_{k \in [W]}\Pr\left[\text{On the $j'$-th query to $\RO^W$, $R_{Q_k}$ contains a value in $X$}\right]\\
            &= \sum_{k \in [W]} \sum_{(i, x) \in X} \Tr\left[\left(\ketbra{i,x}_{R_{Q_k}} \otimes \mathbb{I}_{R \backslash R_{Q_k}}\right) \ketbra{\psi_{j'-1}^\RO}\right]\\
            &= \sum_{(i, x) \in X} \sum_{k \in [W]} \Tr\left[\left(\ketbra{i,x}_{R_{Q_k}} \otimes \mathbb{I}_{R \backslash R_{Q_k}}\right) \ketbra{\psi_{j'-1}^\RO}\right]\\
            &= \sum_{(i, x) \in X} w_{i,x}^{j'}(\RO)
        \end{align*}
        The second line follows from the union bound.

        Next, this implies that the cumulative query weight after $j$ queries obeys a similar relation:
        \begin{align*}
            \sum_{\bfz \in X^{\lor W}} w_{\bfz}^{\leq j}(\RO^W) &= \sum_{\bfz \in X^{\lor W}} \sum_{j' = 0}^j w_{\bfz}^{j'}(\RO^W) = \sum_{j' = 0}^j \sum_{\bfz \in X^{\lor W}}  w_{\bfz}^{j'}(\RO^W)\\
            &\leq \sum_{j' = 0}^j \sum_{(i, x) \in X} w_{i,x}^{j'}(\RO) = \sum_{(i, x) \in X} \sum_{j' = 0}^j w_{i,x}^{j'}(\RO)\\
            &= \sum_{(i, x) \in X} w_{i,x}^{\leq j}(\RO)
        \end{align*}

        Finally, we apply the swapping lemma to $\RO^W$ and $(\RO')^W$. This is allowed because the adversary makes queries of width $1$ to the oracles $\RO^W$ or $(\RO')^W$. Then:
        \begin{align*}
            \left\|\ket{\psi^\RO_{j}} - \ket{\psi^{\RO'}_{j}}\right\|_2 &= \left\|\ket{\psi^{\RO^W}_{j}} - \ket{\psi^{(\RO')^W}_{j}}\right\|_2\\
            &\leq 2 \cdot \sqrt{j \cdot \sum_{\bfz \in X^{\lor W}} w_{\bfz}^{\leq j}(\RO^W)}\\
            &\leq 2 \cdot \sqrt{j \cdot \sum_{(i, x) \in X} w_{i,x}^{\leq j}(\RO)}
        \end{align*}
    \end{proof}

The following lemma says that the Euclidean distance between two states upper-bounds the trace distance.
    \begin{lemma}\label{thm:trace-dist-euclidian-dist}
        For any quantum pure states $\ket{\psi}$ and $\ket{\phi}$, \[\mathsf{TraceDist}\left(\ket{\psi}, \ket{\phi}\right) \leq \left\|\ket{\psi} - \ket{\phi}\right\|_2\]
    \end{lemma}
    \begin{proof}
    \begin{align*}
        \left\|\ket{\psi} - \ket{\phi}\right\|_2^2 &= 
        \left(\bra{\psi} - \bra{\phi}\right)\left(\ket{\psi} - \ket{\phi}\right)\\
        &= \braket{\psi}{\psi} + \braket{\phi}{\phi} - \braket{\phi}{\psi} - \braket{\psi}{\phi}\\
        &= 2 - 2 \cdot \mathsf{Re}\left(\braket{\phi}{\psi}\right)\\
        &\geq 2 - 2 \cdot \abs{\mathsf{Re}\left(\braket{\phi}{\psi}\right)}
        \end{align*}
        
        Next, we will show that $\mathsf{TraceDist}\left(\ket{\psi}, \ket{\phi}\right)^2 \leq 2 - 2 \cdot \abs{\mathsf{Re}\left(\braket{\phi}{\psi}\right)}$.
        \begin{align*}
            \abs{\mathsf{Re}\left(\braket{\phi}{\psi}\right)} &\leq 1\\
            0 &\leq 1 - \abs{\mathsf{Re}\left(\braket{\phi}{\psi}\right)}\\
            0 &\leq \left(1 - \abs{\mathsf{Re}\left(\braket{\phi}{\psi}\right)}\right)^2\\
            &= 1 + \abs{\mathsf{Re}\left(\braket{\phi}{\psi}\right)}^2 - 2 \cdot \abs{\mathsf{Re}\left(\braket{\phi}{\psi}\right)}\\
            1 - \abs{\mathsf{Re}\left(\braket{\phi}{\psi}\right)}^2 &\leq 2 - 2 \cdot \abs{\mathsf{Re}\left(\braket{\phi}{\psi}\right)}\\
            1 - \abs{\braket{\phi}{\psi}}^2 &\leq \\
            \mathsf{TraceDist}\left(\ket{\psi}, \ket{\phi}\right)^2 &= 
        \end{align*}
        Then
        \begin{align*}
            \mathsf{TraceDist}\left(\ket{\psi}, \ket{\phi}\right)^2 &\leq 2 - 2 \cdot \abs{\mathsf{Re}\left(\braket{\phi}{\psi}\right)}\\
            &\leq \left\|\ket{\psi} - \ket{\phi}\right\|_2^2\\
            \mathsf{TraceDist}\left(\ket{\psi}, \ket{\phi}\right) &\leq \left\|\ket{\psi} - \ket{\phi}\right\|_2
    \end{align*}
\end{proof}
\section{Ruling Out $o(\log \secp)$-Depth Strategies}\label{sec:main}
Our main result (\cref{thm:CR-for-depth-bounded-adversaries}) says that the construction in \cref{sec:prelim-YZ-protocol} provides $o(\secp^{c/2})$ bits of min-entropy (for some constant $c$) as long as the adversary is limited to query depth $o(\log \secp)$.

\begin{theorem}[Certifiable Min-Entropy]\label{thm:CR-for-depth-bounded-adversaries}
    The candidate construction of a certifiable min-entropy protocol given in \cref{sec:prelim-YZ-protocol} satisfies $(D, \minent)$-certifiable min-entropy (\cref{def:certifiable-min-entropy}) for any $D(\secp) = o(\log \secp)$ and any $\minent(\secp) = o(\secp^{c/2})$, where $c$ is a constant parameter of the list-recoverable code.
\end{theorem}

The rest of \cref{sec:main} is devoted to proving \cref{thm:CR-for-depth-bounded-adversaries}.\\

Let $(\RO, \bfX)$ be the random variables referring to the random oracle and $\cA$'s output respectively, and let $(\ro, \bfx)$ be generic values that they take. $(\RO, \bfX)$ are jointly sampled from the following distribution. First $\RO \getsr \cH$, and then $\bfX \gets \cA^\RO(1^\secp)$.

\subsection*{The adversary heavily queries their answer.}
First, we prove roughly that if the adversary is able to output a correct answer $\bfx$ with high probability, then they must have given non-negligible query weight to most symbols of $\bfx$. Intuitively, this is because the adversary must check that most symbols of $\bfx$ actually hash to $\mathbf{0}$; otherwise, they would output $\bfx$ even when those symbols hash to $1$.

\begin{lemma}\label{thm:must-heavy-query}
    Let $p$ be any probability $\in (0,1]$, and let $s$ be any value $\in [1,n]$. Let $D, W \geq 1$ be the query depth and width, respectively, of $\cA$, and let $Q = D \cdot W$.
    
    Over the randomness of $(\RO, \bfX)$, the probability is $\leq 2^{-(s-1)}$ that we sample an $(\ro, \bfx)$ that satisfy all of the following conditions:
    \begin{enumerate}
        \item $\cA^\ro$ gives small query weight to many positions of $\bfx$: There exists \emph{no} set $I \subseteq [n]$ of size $\geq n - s$ such that for all $i \in I$,
        \[w_{i, x_i}^{\leq D}(\ro) \geq \frac{p^2}{16 Q n}\]
        \item $\bfx$ has high probability given $\ro$: $\Pr[\bfX = \bfx |\RO = \ro] \geq p$.
        \item $\bfx$ is correct: $\bfx \in C$ and $\ro(\bfx) = \mathbf{0}$.
    \end{enumerate}
\end{lemma}
\begin{proof}
$ $
    \paragraph{Constructing Good and Bad Sets:} Let $S$ be the set of all pairs $(\ro, \bfx)$ that satisfy the conditions of \cref{thm:must-heavy-query} ($S$ is the ``good'' set). $\Pr[S]$ is the probability that all the conditions of \cref{thm:must-heavy-query} are satisfied, over the randomness of sampling $\ro \gets \cH$ and $\bfx \gets \cA^\ro$. We will show that $\Pr[S] \leq 2^{-(s-1)}$. 

    Next, for each $(\ro,\bfx) \in S$, let us construct a set $T_{\ro, \bfx}$ of ``bad'' pairs $(\ro', \bfx)$ on which $\bfx$ is not correct. 
    
    Condition 1 implies that for any set $I \subseteq [n]$ of size $\geq n - s$, there exists an $i \in I$ such that \[w_{i, x_i}^{\leq D}(\ro) < \frac{p^2}{16 Q n}\]
    
    That means there exists another set $I_2 \subseteq [n]$ of size $|I_2| \geq s + 1$ such that for all $i \in I_2$, 
    \[w_{i, x_i}^{\leq D}(\ro) < \frac{p^2}{16 Q n}\]
    
    To construct $T_{\ro, \bfx}$, choose any non-empty subset of $\{(i, x_i)\}_{i \in I_2}$ and reprogram $\ro$ on these inputs, flipping the output from $0$ to $1$. Let us call the new oracle $\ro'$. Let $T_{\ro, \bfx}$ be the set of all pairs $(\ro', \bfx)$ that can be constructed in this way. 
    
    \paragraph{Properties:} First
    \[|T_{\ro, \bfx}| = 2^{|I_2|} - 1 \geq 2^{s + 1} - 1 \geq 2^{s}\]

    Second, for any $(\ro', \bfx) \in T_{\ro, \bfx}$, $\bfx$ is not a correct answer for oracle $\ro'$ because $\ro'(\bfx) \neq \mathbf{0}$.

    Third, the sets $T_{\ro, \bfx}$ are disjoint. That is to say, for any two distinct pairs $(h_1, \bfx_1)$ and $(h_2, \bfx_2)$ in $S$, 
    \[T_{h_1, \bfx_1} \cap T_{h_2, \bfx_2} = \emptyset\]
    For any $(\ro', \bfx') \in T_{\ro, \bfx}$, $\bfx = \bfx'$, and $\ro$ can be computed from $(\ro', \bfx')$ by reprogramming $\ro'$ to map $\bfx'$ to $\mathbf{0}$. This procedure computes $(\ro, \bfx)$ from any $(\ro', \bfx') \in T_{\ro, \bfx}$, so the sets $T_{\ro, \bfx}$ must be disjoint.

    Fourth, 
    \begin{claim}
        For any $(\ro', \bfx) \in T_{\ro, \bfx}$,
        \[\frac{1}{2} \cdot \Pr[\bfX = \bfx|\RO = \ro] \leq \Pr[\bfX = \bfx|\RO = \ro']\]
    \end{claim}
    \begin{proof}
    This follows from the swapping lemma (\cref{thm:swapping-lemma}) and \cref{thm:trace-dist-euclidian-dist}. The adversary samples their output $\bfx$ by applying a measurement $M$ to their final state (either $\ket{\psi_D^\ro}$ or $\ket{\psi_D^{\ro'}}$).
    \begin{align*}
        \Pr[\bfX = \bfx|\RO = \ro] - \Pr[\bfX = \bfx|\RO = \ro'] &\leq \mathsf{TraceDist}\left[\ket{\psi_D^{\ro'}}, \ket{\psi_D^{\ro}}\right]\\
        &\leq \left\|\ket{\psi_D^{\ro'}} -  \ket{\psi_D^{\ro}}\right\|_2\quad \text{(\cref{thm:trace-dist-euclidian-dist})}\\
        &\leq 2 \cdot \sqrt{D \cdot \sum_{\substack{(i, x_i) : \\\ro(i, x_i) \neq \ro'(i, x_i)}} w_{i, x_i}^{\leq D}(\ro)}\quad \text{(\cref{thm:swapping-lemma})}\\
        &\leq 2 \cdot \sqrt{D \cdot \sum_{i \in I_2} w_{i, x_i}^{\leq D}(\ro)}\\
        &\leq 2 \cdot \sqrt{D \cdot \sum_{i \in I_2} \frac{p^2}{16 Q n}}\\
        &= 2 \cdot \sqrt{\frac{D|I_2| p^2}{16Qn}}\\
        &\leq 2 \cdot \sqrt{\frac{Q n p^2}{16Qn}}\\
        &= \frac{p}{2}\\
        \Pr[\bfX = \bfx|\RO = \ro] - \frac{p}{2}&\leq \Pr[\bfX = \bfx|\RO = \ro']\\
        \Pr[\bfX = \bfx|\RO = \ro] - \frac{1}{2} \cdot \Pr[\bfX = \bfx|\RO = \ro]&\leq\\
        \frac{1}{2} \cdot \Pr[\bfX = \bfx|\RO = \ro] &\leq \Pr[\bfX = \bfx|\RO = \ro']
    \end{align*}
    We used the fact that $\frac{p}{2} \leq \frac{1}{2} \cdot \Pr[\bfX = \bfx|\RO = \ro]$.
    \end{proof}

    \paragraph{Finishing the proof:} Now we will show that $\Pr[S]\leq 2^{-(s-1)}$. 
    
    Let $\Pr[T_{\ro, \bfx}]$ be over the randomness of sampling $\ro' \gets \cH$ and $\bfx \gets \cA^{\ro'}$ given $\ro'$. Since the sets $\left(T_{\ro, \bfx}\right)_{(\ro, \bfx) \in S}$ are mutually disjoint,
    \begin{align*}
        1 &\geq \sum_{(\ro, \bfx) \in S} \Pr[T_{\ro, \bfx}]\\
        &= \sum_{(\ro, \bfx) \in S} \sum_{(\ro', \bfx) \in T_{\ro, \bfx}} \frac{1}{|\cH|} \cdot \Pr[\bfX = \bfx|\RO = \ro']\\
        &\geq \sum_{(\ro, \bfx) \in S} \sum_{(\ro', \bfx) \in T_{\ro, \bfx}} \frac{1}{2 |\cH|} \cdot \Pr[\bfX = \bfx|\RO = \ro]\\
        &= \sum_{(\ro, \bfx) \in S} |T_{\ro, \bfx}| \cdot \frac{1}{2|\cH|} \cdot \Pr[\bfX = \bfx|\RO = \ro]\\
        &\geq \sum_{(\ro, \bfx) \in S} 2^{s} \cdot \frac{1}{2 |\cH|} \cdot \Pr[\bfX = \bfx|\RO = \ro]\\
        &= 2^{s - 1} \cdot \sum_{(\ro, \bfx) \in S} \frac{1}{|\cH|} \cdot \Pr[\bfX = \bfx|\RO = \ro]\\
        &= 2^{s - 1} \cdot \Pr[S]\\
        2^{-(s - 1)} &\geq \Pr[S]
    \end{align*}

    We've shown that the conditions of \cref{thm:must-heavy-query} are satisfied with probability $\leq 2^{-(s - 1)}$.
\end{proof}

\subsection*{Low-depth adversaries cannot heavily query a correct answer.}
Here we prove that if the adversary's query depth is $o(\log \secp)$, then with overwhelming probability over $\RO$, $\cA^\RO$ does not \textit{heavily query} any correct answer $\bfx$. Heavy querying means that many symbols of $\bfx$ each receive large query weight from $\cA^\RO$.

\begin{lemma}\label{thm:2-query-security}
    Let $p(\secp)$ be any function satisfying $p(\secp) = \omega\left(2^{-\frac{1}{2} \cdot \left(\secp^{c/2}\right)}\right)$. Let $D(\secp)$ be any function satisfying $1 \leq D(\secp) = o\left(\log \secp\right)$. Let $W(\secp)$ be any function satisfying $1 \leq W(\secp)  = O(\poly(\secp))$, and let $s = \lfloor \zeta n \rfloor$. Let $\cA$ have query width $W$ and query depth $D$, making $Q = W \cdot D$ queries in total.
    
    Then with overwhelming probability over the randomness of $\RO$, $\RO$ satisfies the following condition: 
    \begin{enumerate}
        \item For all $\bfx$ such that $\bfx$ is correct ($\bfx \in C$ and $\RO(\bfx) = \mathbf{0}$), there exists \textit{no} set $I \subseteq [n]$ of size $\geq n - s$ such that for all $i \in I$,
        \[w_{i, x_i}^{\leq D}(\RO) \geq \frac{p^2}{16 Q n}\]
    \end{enumerate}
\end{lemma}

\begin{proof}
$ $\\\\
    \noindent\textbf{Query Weight Thresholds:}
    Let us define a threshold weight given to a symbol after $q$ query layers. For each $q \in \{0, 1, \dots, D\}$, let the threshold be
    \[t_q = \left(\frac{2 Q n}{\ell}\right)^{\left(4^{-q}\right)}\]

    \begin{claim}\label{thm:bounds-on-t-q}
        For any $q \in \{0, 1, \dots, D\}$ and sufficiently large $\secp$, 
        \[\frac{Q}{\ell} \leq t_q \leq 1\]
    \end{claim}
    \begin{proof}
    We know that
    \begin{align*}
        \ell &= 2^{\left(\secp^c\right)}\\
        Q &= W \cdot D = O(\poly(\secp)) \\
        n &= 2^{\lfloor \log \secp \rfloor} - 1 \leq \secp\\
        1 &\leq n \quad \text{for sufficiently large $\secp$}
    \end{align*}
    Therefore,
    \begin{align*}
        2 Q n &\leq \ell \quad \text{for sufficiently large $\secp$}\\
        \frac{2 Q n}{\ell} &\leq 1\\
    \end{align*}
    Additionally, $4^{-q} \in (0, 1]$. Then
    \begin{align*}
        \frac{2 Q n}{\ell} &\leq \left(\frac{2 Q n}{\ell}\right)^{\left(4^{-q}\right)} \leq 1\\
        \frac{2 Q n}{\ell} &\leq t_q \leq 1\\
        \frac{Q}{\ell} &\leq t_q \leq 1
    \end{align*}
    \end{proof}

    \begin{claim}\label{thm:upper-bound-on-t-q}
        For any $q \in \{0, 1, \dots, D\}$ and sufficiently large $\secp$,
        \[t_q \leq \frac{p^2}{16 Q n}\]
    \end{claim}
    \begin{proof}
    First,
    \begin{align*}
        q &\leq D \leq \frac{1}{4} \cdot \left(c \log \secp - 2\right) \quad \text{for sufficiently large $\secp$}\\
        -q &\geq \frac{1}{4} \cdot \left(2 - c \log \secp\right)\\
        4^{-q} &\geq 4^{\frac{1}{4} \cdot \left(2 - c \log \secp\right)}\\
        &= 2^{\frac{2}{4} \cdot \left(2 - c \log \secp\right)}\\
        &= 2^{\left(1 - \frac{c}{2} \cdot \log \secp\right)}\\
        &= 2 \cdot 2^{- \frac{c}{2} \cdot \log \secp}\\
        &= 2 \cdot \secp^{-\frac{c}{2}}
    \end{align*}

    Second,
    \begin{align*}
        p(\secp) &= \omega\left(2^{-\frac{1}{2} \cdot \left(\secp^{c/2}\right)}\right)\\
        p(\secp) &\geq 2^{-\frac{1}{2} \cdot \left(\secp^{c/2}\right)} \quad \text{for sufficiently large $\secp$}\\
        p^2(\secp) &\geq 2^{-\left(\secp^{c/2}\right)}
    \end{align*}

    Third, the following are true for sufficiently large $\secp$:
    \begin{equation}\label{eq:true-when-secp-is-large}
    \begin{split}
        \frac{2 Q n}{\ell} &\leq 1\\
        1 &\leq 2 Q n\\
        2 \cdot \secp^{-\frac{c}{2}} &\leq 1\\
        \frac{32 Q^2 n^2}{2^{\left(\secp^{\frac{c}{2}}\right)}} &\leq 1
    \end{split}
    \end{equation}
    Let us assume that $\secp$ is large enough that all of the conditions in \cref{eq:true-when-secp-is-large} are satisfied.
    
    Fourth, 
    \begin{align*}
        \frac{2 Q n}{\ell} &\leq 1\\
        \left(\frac{2 Q n}{\ell}\right)^{\left(4^{-q}\right)} &\leq \left(\frac{2 Q n}{\ell}\right)^{\left(2 \cdot \secp^{-\frac{c}{2}}\right)}\\
        t_q &\leq \left(2 Q n\right)^{\left(2 \cdot \secp^{-\frac{c}{2}}\right)} \cdot \ell^{\left(-2 \cdot \secp^{-\frac{c}{2}}\right)}\\
        &\leq \left(2 Q n\right) \cdot \ell^{\left(-2 \cdot \secp^{-\frac{c}{2}}\right)}\\
        &= \left(2 Q n\right) \cdot \left(2^{\left(\secp^c\right)}\right)^{\left(-2 \cdot \secp^{-\frac{c}{2}}\right)}\\
        &= \left(2 Q n\right) \cdot 2^{\left(\secp^c\right) \cdot \left(-2\right) \cdot \left(\secp^{-\frac{c}{2}}\right)}\\
        &= \left(2 Q n\right) \cdot 2^{-2 \cdot \left(\secp^{\frac{c}{2}}\right)}\\
        &= \left(2 Q n\right) \cdot \left(2^{-\left(\secp^{\frac{c}{2}}\right)}\right) \cdot \left(16 Q n\right) \cdot \frac{1}{16 Q n} \cdot \left(2^{-\left(\secp^{\frac{c}{2}}\right)}\right)\\
        &\leq \left(2 Q n\right) \cdot \left(2^{-\left(\secp^{\frac{c}{2}}\right)}\right) \cdot \left(16 Q n\right) \cdot \frac{1}{16 Q n} \cdot p^2\\
        &\leq \frac{32 Q^2 n^2}{2^{\left(\secp^{\frac{c}{2}}\right)}} \cdot \frac{p^2}{16 Q n}\\
        &\leq \frac{p^2}{16 Q n} \quad \text{for sufficiently large $\secp$}
    \end{align*}
    \end{proof}

    \begin{claim}\label{thm:lower-bound-t-q}
        For any $q \in \{1, 2, \dots, D\}$ and sufficiently large $\secp$, 
        \[\frac{Q}{\ell} + 2Q \cdot \sqrt{D n \cdot t_{q-1}} \leq t_q\]
    \end{claim}
    \begin{proof}
    First,
    \begin{align*}
        D &\leq \frac{1}{4} \cdot \left(c \cdot \log \secp - 2\right) \quad \text{for sufficiently large $\secp$}\\
        q &\leq D \leq \frac{1}{4} \cdot \left(c \cdot \log \secp - 2\right)\\
        2q &\leq \frac{1}{2} \cdot \left(c \cdot \log \secp - 2\right) = \frac{c}{2} \cdot \log \secp - 1\\
        2q - 1 &\leq \frac{c}{2} \cdot \log \secp - 2\\
        1 - 2q &\geq 2 - \frac{c}{2} \cdot \log \secp\\
        2^{1 - 2q} &\geq 2^{2 - \frac{c}{2} \cdot \log \secp}\\
        &= 4 \cdot \secp^{-c/2}
    \end{align*}

    Second,
    
    \begin{align*}
        \frac{\ell}{2 Q n} &= 2^{\left(\secp^c - O(\log \secp)\right)}\\
        \left(\frac{\ell}{2 Q n}\right)^{4 \cdot \secp^{-c/2}} &= 2^{\left(\secp^c - O(\log \secp)\right) \cdot \left(4 \cdot \secp^{-c/2}\right)}\\
        &= 2^{\left(4 \cdot \secp^{c/2} - O(\secp^{-c/2} \cdot \log \secp)\right)}\\
        &= 2^{\Theta\left(\secp^{c/2}\right)}\\
        &\geq 9 W^2 D^3 n \quad \text{for sufficiently large $\secp$}
    \end{align*}

    Third,
    \begin{align*}
        \left(\frac{\ell}{2 Q n}\right)^{\left(2^{1-2q}\right)} &= \left(\frac{2 Q n}{\ell}\right)^{\left(-1 \cdot 2^{1-2q}\right)} = \left(\frac{2 Q n}{\ell}\right)^{\left(-2 \cdot 4^{-q}\right)}\\
        &= \left(\frac{2 Q n}{\ell}\right)^{\left[(2 - 4) \cdot 4^{-q}\right]} = \left(\frac{2 Q n}{\ell}\right)^{\left[2 \cdot 4^{-q} - 4^{-(q-1)}\right]}\\
        &= \frac{\left[\left(\frac{2 Q n}{\ell}\right)^{\left(4^{-q}\right)}\right]^2}{\left(\frac{2 Q n}{\ell}\right)^{\left(4^{-(q-1)}\right)}} = \frac{t_q^2}{t_{q-1}}
    \end{align*}

    Finally,
    \begin{align*}
        \ell &\geq 2 Q n \quad \text{for sufficiently large $\secp$}\\
        \frac{\ell}{2 Q n} &\geq 1\\
        \left(\frac{\ell}{2 Q n}\right)^{\left(2^{1-2q}\right)} &\geq \left(\frac{\ell}{2 Q n}\right)^{4 \cdot \secp^{-c/2}}\\
        \frac{t_q^2}{t_{q-1}} &\geq 9W^2D^3n\\
        t_q &\geq 3 DW \sqrt{D n \cdot t_{q-1}} = 3 Q \sqrt{D n \cdot t_{q-1}}\\
        &= Q \sqrt{D n} \cdot \sqrt{t_{q-1}} + 2Q \sqrt{D n \cdot t_{q-1}}\\
        &\geq \sqrt{t_{q-1}} + 2Q \sqrt{D n \cdot t_{q-1}} \quad \text{because $D,Q, n \geq 1$ for sufficiently large $\secp$}\\
        &\geq t_{q-1} + 2Q \sqrt{D n \cdot t_{q-1}} \quad \text{because $t_{q-1} \leq 1$ for sufficiently large $\secp$ (\cref{thm:bounds-on-t-q})}\\
        &\geq \frac{Q}{\ell} + 2Q \sqrt{D n \cdot t_{q-1}} \quad \text{because $t_{q-1} \geq \frac{Q}{\ell}$ for sufficiently large $\secp$ (\cref{thm:bounds-on-t-q})}
    \end{align*}
    \end{proof}

    \paragraph{Constructing the Good Set:} Let us construct a set $S$ of ``good'' oracles $\ro$. For every $\ro \in \cH$ that violates the condition of \cref{thm:2-query-security}, pick an $\bfx$ such that: 
    \begin{enumerate}
        \item $\bfx$ is correct ($\bfx \in C$ and $\ro(\bfx) = \mathbf{0}$), and
        \item there exists a set $I \subseteq [n]$ of size $\geq n - s$ such that for all $i \in I$,
        \[w_{i, x_i}^{\leq D}(\ro) \geq \frac{p^2}{16 Q n}\]
    \end{enumerate}
    Such an $\bfx$ must exist. Then add this $(\ro, \bfx)$ pair to $S$. Note that no two entries of $S$ have the same $\ro$-value, so the size of $S$ is the number of $\ro$-values that violate the condition of \cref{thm:2-query-security}. Our goal is to show that $\frac{|S|}{|\cH|} = \negl(\secp)$. This implies that with overwhelming probability over the choice of $\ro$, the condition of \cref{thm:2-query-security} is satisfied.

    \paragraph{Constructing the Bad Set:} For any $(\ro, \bfx) \in S$, we will construct a set of ``bad'' oracles $T_{\ro,\bfx}$ where $\bfx$ is no longer correct, but $\cA$ still wastes query weight on $\bfx$. 

    First, let us say that \textbf{$\cA^\ro$ heavily queries $\bfx$} if there exists a set $I_1 \subseteq [n]$ of size $|I_1| \geq n - s$ such that for all $i \in I_1$,
    \[w_{i, x_i}^{\leq D}(\ro) \geq \frac{Q}{\ell}\]
    
    Additionally, let us say that \textbf{$\cA^\ro$ strongly-heavily queries $\bfx$ after $q$ queries} if there exists a set $I_1 \subseteq [n]$ of size $|I_1| \geq n-s$ such that for all $i \in I_1$,
    \[w_{i, x_i}^{\leq q}(\ro) \geq t_q\]
    
    Second, let $q^*$ be the first query in $[D]$ for which $\cA^\ro$ strongly-heavily queries $\bfx$ after $q^*$ queries.
    
    Such a $q^*$ must exist. Since $(\ro, \bfx) \in S$, there exists a set $I_1 \subseteq [n]$ of size $\geq n - s$ such that for all $i \in I_1$,
    \[w_{i, x_i}^{\leq D}(\ro) \geq \frac{p^2}{16 Q n} \geq t_D\]
    So $\cA^\ro$ strongly-heavily queries $\bfx$ after $D$ queries.

    Third, let $I_2 \subseteq [n]$ be the set of indices of $\bfx$ of size $|I_2| = s + 1$ that receive the smallest weight on the first $q^*-1$ queries. Formally, we require that for all $i \in I_2$ and all $i' \in [n] \backslash I_2$,
    \[w_{i, x_i}^{\leq q^*-1}(\ro) \leq w_{i', x_{i'}}^{\leq q^*-1}(\ro)\]

    Fourth, let us construct our bad set $T_{\ro, \bfx}$ as follows. Choose any non-empty subset of $\{(i, x_i)\}_{i \in I_2}$, and reprogram $\ro$ on this set, flipping the outputs from $0$ to $1$. Call the new oracle $\ro'$, and add $\ro'$ to $T_{\ro, \bfx}$. 
    
    \paragraph{Properties:} Let's state some useful properties of $T_{\ro, \bfx}$ and $I_2$.
    
    First, note that for any $\ro' \in T_{\ro, \bfx}$, $\ro'(\bfx) \neq \mathbf{0}$. Second,
    \[|T_{\ro,\bfx}| = 2^{|I_2|} - 1 = 2^{s + 1} - 1 \geq 2^{s}\]
    
    Third, before the $q^*$-th query, every reprogrammed symbol received small cumulative query weight.
    \begin{lemma}\label{thm:I-2-has-small-query-weight}
        For all $(\ro, \bfx) \in S$, all $i \in I_2$, and all $q \in [q^*]$,
        \begin{align*}
            w_{i, x_i}^{\leq q-1}(\ro) < t_{q^*-1}
        \end{align*}
    \end{lemma}
    \begin{proof}
        We chose $q^*$ to be the \textit{first} query in $[D]$ after which $\cA^\ro$ strongly-heavily queries $\bfx$. Then $\cA^\ro$ does not strongly-heavily query $\bfx$ after $q^* - 1$ queries. That means there does not exist a set $I_1 \subseteq [n]$ of size $|I_1| \geq n - s$ such that for all $i \in [I_1]$:
        \begin{equation}\label{eq:large-query-weight}
            w_{i, x_i}^{\leq q^*-1}(\ro) \geq t_{q^* - 1}
        \end{equation}
        Then the number of $i$-values in $[n]$ for which \cref{eq:large-query-weight} is satisfied is $\leq n - s - 1$. The number of $i$-values for which \cref{eq:large-query-weight} is not satisfied is
        \[\geq n - \left(n - s - 1\right) = s + 1 \] 
        
        Since $I_2$ contains the $s + 1$ indices with the smallest values of $w_{i, x_i}^{\leq q^*-1}(\ro)$, that means that for any $i \in I_2$,
        \[w_{i, x_i}^{\leq q^*-1}(\ro) < t_{q^* - 1}\]
        Since $q \leq q^*$,
        \begin{align*}
            w_{i, x_i}^{\leq q-1}(\ro) &\leq w_{i, x_i}^{\leq q^*-1}(\ro)\\
            &< t_{q^* - 1}
        \end{align*}
    \end{proof}

    Fourth, even if we replace a good $\ro$ with a bad $\ro' \in T_{\ro, \bfx}$, the adversary will still heavily query $\bfx$.
    \begin{lemma}
        For any $(\ro, \bfx) \in S$, any $\ro' \in T_{\ro, \bfx}$, and sufficiently large $\secp$, $\cA^{\ro'}$ heavily queries $\bfx$.
    \end{lemma}
    \begin{proof}
        Let $I_1, I_2$ be the sets defined in the construction of $T_{\ro, \bfx}$. That's to say, $I_1$ is a set $\subseteq [n]$ of size $|I_1| \geq n - s$ such that for all $i \in I_1$, 
        \[w_{i, x_i}^{\leq q^*}(\ro) \geq t_{q^*}\]
        And $I_2 \subseteq [n]$ is a set of size $|I_2| = s+1$ such that for all $i \in I_2$
        \[w_{i, x_i}^{\leq q^*-1}(\ro) < t_{q^*-1}\]

        Next, for any $q \in [q^*]$, any $\ro' \in T_{\ro, \bfx}$, and any $i \in I_1$:
        \begin{align*}
            w_{i, x_i}^{q}(\ro) - w_{i, x_i}^{q}(\ro') &= \sum_{k \in [W]} \Tr\left[\left(\ketbra{i,x_i}_{R_{Q_k}} \otimes \mathbb{I}_{R \backslash R_{Q_k}}\right) \left(\ketbra{\psi_{q-1}^\ro} - \ketbra{\psi_{q-1}^{\ro'}}\right)\right]\\
            &\leq \sum_{k \in [W]} \mathsf{TraceDist}\left(\ket{\psi^{\ro}_{q-1}},\ket{\psi^{\ro'}_{q-1}}\right) = W \cdot \mathsf{TraceDist}\left(\ket{\psi^{\ro}_{q-1}},\ket{\psi^{\ro'}_{q-1}}\right)\\
            &\leq W \cdot \left\|\ket{\psi^{\ro}_{q-1}} - \ket{\psi^{\ro'}_{q-1}}\right\|_2 \quad \text{(\cref{thm:trace-dist-euclidian-dist})}\\
            &\leq 2W \cdot \sqrt{(q-1) \cdot \sum_{\substack{(i',x'):\\
            \ro(i',x') \neq \ro'(i',x')}} w_{i',x'}^{\leq q-1}(\ro)} \quad \text{(\cref{thm:swapping-lemma})}\\
            &\leq 2W \cdot \sqrt{(q-1) \cdot \sum_{i \in I_2} w_{i,x_i}^{\leq q-1}(\ro)}\\
            &\leq 2W \cdot \sqrt{(q-1) \cdot \sum_{i \in I_2} t_{q^*-1}} \quad \text{(\cref{thm:I-2-has-small-query-weight})}\\
            &=  2W \cdot \sqrt{(q-1) \cdot |I_2| \cdot t_{q^*-1}}\\
            &\leq 2W \cdot \sqrt{(q-1) \cdot n \cdot t_{q^*-1}}\\
            &\leq 2W \cdot \sqrt{D n \cdot t_{q^*-1}}\\
            w_{i, x_i}^{q}(\ro) - 2W \cdot \sqrt{D n \cdot t_{q^*-1}} &\leq w_{i, x_i}^{q}(\ro')
        \end{align*}
        
        Now let us sum over all $q \in [q^*]$.
        \begin{align*}
            \sum_{q = 1}^{q^*} w_{i, x_i}^{q}(\ro') &\geq \sum_{q=1}^{q^*} \left(w_{i, x_i}^{q}(\ro) - 2W \cdot \sqrt{D n \cdot t_{q^*-1}}\right)\\
            w_{i, x_i}^{\leq q^*}(\ro') &\geq w_{i, x_i}^{\leq q^*}(\ro) - q^* \cdot 2W \cdot \sqrt{D n \cdot t_{q^*-1}}\\
            &\geq t_{q^*} - q^* \cdot 2W \cdot \sqrt{D n \cdot t_{q^*-1}}\\
            &\geq t_{q^*} - D \cdot 2W \cdot \sqrt{D n \cdot t_{q^*-1}}\\
            &= t_{q^*} - 2Q \cdot \sqrt{D n \cdot t_{q^*-1}}\\
            &\geq \frac{Q}{\ell} + 2Q \cdot \sqrt{D n \cdot t_{q^*-1}} - 2Q \cdot \sqrt{D n \cdot t_{q^*-1}} \quad \text{(\cref{thm:lower-bound-t-q})}\\
            &= \frac{Q}{\ell} \\
            w_{i, x_i}^{\leq D}(\ro') &\geq \frac{Q}{\ell}
        \end{align*}
        Therefore, $\cA^{\ro'}$ heavily queries $\bfx$.
    \end{proof}

    \begin{claim}
        The number of codewords that $\cA^{\ro'}$ heavily queries is $\leq L := 2^{\tilde{O}(\secp^{c'})}$.
    \end{claim}
    \begin{proof}
        The total cumulative query weight given to all inputs is $\leq Q = D \cdot W$ (\cref{thm:total-cumulative-query-weight}). The number of inputs $(i, x)$ with cumulative query weight $w_{i, x}^{\leq D}(\ro') \geq \frac{Q}{\ell}$ is 
        \[\leq Q \cdot \frac{\ell}{Q} = \ell\] 
        If $\cA^{\ro'}$ heavily queries $\bfx$ then at least $n - s$ coordinates of $\bfx$ belong to the list of inputs that have cumulative query weight $\geq \frac{Q}{\ell}$. Note that 
        \[n-s \geq n - \zeta \cdot n = (1 - \zeta) \cdot n\]
        Then the list recoverability property (\cref{thm:list-recoverability}) says that the number of codewords $\bfx$ that satisfy this property is $\leq L$.
    \end{proof}

    \begin{claim}
        Any $\ro'$ belongs to $T_{\ro, \bfx}$ for at most $L$ values of $(\ro, \bfx) \in S$.
    \end{claim}
    \begin{proof}
        Assume toward contradiction that there are $L+1$ distinct values of $(\ro, \bfx) \in S$ such that $\ro' \in T_{\ro, \bfx}$. Call them $(h_j, \bfx_j)_{j \in [L+1]}$. We know that $\cA^{\ro'}$ heavily queries each $\bfx_j$, but there are at most $L$ codewords that $\cA^{\ro'}$ heavily queries. By the pigeonhole principle, there are two different values $j, j' \in [L+1]$ such that $\bfx_j = \bfx_{j'}$.
        
        Furthermore, $h_{j}$ can be constructed from $\ro'$ and $\bfx_{j}$ by reprogramming $\ro'$ to map $\bfx_{j}$ to $\mathbf{0}$. By the same procedure, $h_{j'}$ can be constructed from $\ro'$ and $\bfx_{j'}$. Since $\bfx_j = \bfx_{j'}$, this implies that $h_j = h_{j'}$ because they are constructed from $(\ro', \bfx_j)$ by the same procedure. In summary $(h_j, \bfx_j) = (h_{j'}, \bfx_{j'})$, which contradicts the claim that each $(h_{j''}, \bfx_{j''})$ is distinct. Therefore the initial assumption must be false, so in fact there are at most $L$ values of $(\ro,\bfx) \in S$ for which $\ro' \in T_{\ro, \bfx}$.
    \end{proof}

    Now we can finish the proof. Since each $\ro' \in \cH$ belongs to at most $L$ different sets $T_{\ro, \bfx}$,
    \begin{align*}
        |\cH| \cdot L &\geq \sum_{\ro' \in \cH} \sum_{(\ro, \bfx) \in S}  \mathbbm{1}_{\ro' \in T_{\ro, \bfx}}\\
        &= \sum_{(\ro, \bfx) \in S} \sum_{\ro' \in \cH} \mathbbm{1}_{\ro' \in T_{\ro, \bfx}}\\
        &= \sum_{(\ro, \bfx) \in S} |T_{\ro, \bfx}|\\
        &\geq \sum_{(\ro, \bfx) \in S} 2^{s}\\
        &= |S| \cdot 2^{s} = |S| \cdot 2^{\lfloor \zeta n \rfloor}\\
        &= |S| \cdot 2^{\Omega(\secp)}\\
        \frac{L}{2^{\Omega(\secp)}} &\geq \frac{|S|}{|\cH|}\\
        2^{\tilde{O}(\secp^{c'}) - \Omega(\secp)} &=\\
        \negl(\secp) &=
    \end{align*}
    $2^{\tilde{O}(\secp^{c'}) - \Omega(\secp)} = \negl(\secp)$ because $c' < 1$.

    We've shown that $\frac{|S|}{|\cH|}$ is $\negl(\secp)$, which completes the proof.    
\end{proof}

\subsection*{Finishing the proof}

Let us assume toward contradiction that \cref{thm:CR-for-depth-bounded-adversaries} is false. Then there is some $D = o(\log \secp)$, some $\minent = o(\secp^{c/2})$, some $W = \poly(\secp)$, some inverse-polynomial $\delta(\secp)$, and some adversary $\cA$ making $D$ layers of $W$ parallel queries such that 
\begin{align*}
    \Pr_{\RO}\left[\Pr\left[\Verify^\RO [1^\secp, 1^\minent, \cA^\RO(1^\secp, 1^\minent)] \neq \bot\right] \geq \delta(\secp) \land \Minent\left(\cA^\RO_\top (1^\secp, 1^\minent)\right) \leq \minent(\secp)\right] = \text{non-negl}(\secp).
\end{align*}

Let us define two parameters.
\begin{align*}
    \text{Let } p(\secp) &= \delta(\secp) \cdot 2^{-\minent(\secp)}\\
    s(\secp) &= \lfloor \zeta n \rfloor
\end{align*}

Let us also define two events based on $\RO$:
\begin{itemize}
    \item Let $\text{Event}_1$ be the event that $\RO$ satisfies:
    \[\Pr\left[\Verify^\RO [1^\secp, 1^\minent, \cA^\RO(1^\secp, 1^\minent)] \neq \bot\right] \geq \delta(\secp) \land \Minent\left(\cA^\RO_\top (1^\secp, 1^\minent)\right) \leq \minent(\secp)\]
    \item Let $\text{Event}_2$ be the event that $\RO$ satisfies the following condition:

    For all $\bfx$ such that $\bfx$ is correct ($\bfx \in C$ and $\RO(\bfx) = \mathbf{0}$), there exists \textit{no} set $I \subseteq [n]$ of size $\geq n - s$ such that for all $i \in I$,
        \[w_{i, x_i}^{\leq D}(\RO) \geq \frac{p^2}{16 Q n}\]
\end{itemize}
We have that $\Pr_\RO[\text{Event}_1] = \text{non-negl}(\secp)$ if the adversary breaks the min-entropy property.

Finally, let us define $\text{Event}_3$ based on $(\RO, \bfX)$ as the event that the conditions of \cref{thm:must-heavy-query} are satisfied.
\begin{itemize}
    \item Let $\text{Event}_3$ be the event that $(\RO, \bfX)$ have values $(\ro, \bfx)$ that satisfy:
    \begin{enumerate}
        \item $\cA^\ro$ gives small query weight to many positions of $\bfx$: There exists \emph{no} set $I \subseteq [n]$ of size $\geq n - s$ such that for all $i \in I$,
        \[w_{i, x_i}^{\leq D}(\ro) \geq \frac{p^2}{16 Q n}\]
        \item $\bfx$ has high probability given $\ro$: $\Pr[\bfX = \bfx |\RO = \ro] \geq p$.
        \item $\bfx$ is correct: $\bfx \in C$ and $\ro(\bfx) = \mathbf{0}$.
    \end{enumerate}
\end{itemize}

\begin{lemma}
    Let $\ro$ be any value of $\RO$ that satisfies $\text{Event}_1$ and $\text{Event}_2$. Then there exists an $\bfx^\ro$ such that $(\ro, \bfx^\ro)$ satisfy $\text{Event}_3$.
\end{lemma}
\begin{proof}
The first condition of $\text{Event}_1$ is $\Pr\left[\Verify^\RO [1^\secp, 1^\minent, \cA^\RO(1^\secp, 1^\minent)] \neq \bot\right] \geq \delta(\secp)$, which is equivalent to:
    \[\Pr_\bfX[\bfX \in C \land \RO(\bfX) = \mathbf{0} | \RO = \ro] \geq \delta(\secp).\]
Next, 
\begin{align*}
    \Minent\left(\cA^\RO_\top (1^\secp, 1^\minent)\right) &= - \log\left(\max_{\bfx} \Pr_\bfX\left[\bfX = \bfx | \RO = \ro, (\bfX \in C \land \RO(\bfX) = \mathbf{0})\right]\right)
\end{align*}
Furthermore, for any $\bfx$ that is correct ($\bfx \in C \land \ro(\bfx) = \mathbf{0}$),
\begin{align*}
    \Pr_\bfX\left[\bfX = \bfx | \RO = \ro, (\bfX \in C \land \RO(\bfX) = \mathbf{0})\right] &= \frac{\Pr_\bfX\left[\bfX = \bfx, (\bfX \in C \land \RO(\bfX) = \mathbf{0}) | \RO = \ro\right]}{\Pr_\bfX\left[\bfX \in C \land \RO(\bfX) = \mathbf{0}|\RO = \ro\right]}\\
    &= \frac{\Pr_\bfX\left[\bfX = \bfx | \RO = \ro\right]}{\Pr_\bfX\left[\bfX \in C \land \RO(\bfX) = \mathbf{0}|\RO = \ro\right]}\\
    &\leq \frac{\Pr_\bfX\left[\bfX = \bfx | \RO = \ro\right]}{\delta(\secp)}
\end{align*}

The second condition of $\text{Event}_1$ is $\Minent\left(\cA^\RO_\top (1^\secp, 1^\minent)\right) \leq \minent(\secp)$, which implies that:
\begin{align*}
    -\log\left(\max_{\bfx} \Pr_\bfX[\bfX = \bfx | \RO = \ro, (\bfX \in C \land \RO(\bfX) = \mathbf{0})]\right) &\leq \minent(\secp)\\
    \log\left(\max_{\bfx} \Pr_\bfX[\bfX = \bfx | \RO = \ro, (\bfX \in C \land \RO(\bfX) = \mathbf{0})]\right) &\geq -\minent(\secp)\\
    \max_{\bfx} \Pr_{\bfX}[\bfX = \bfx | \RO = \ro, (\bfX \in C \land \RO(\bfX) = \mathbf{0})] &\geq 2^{-\minent(\secp)}
\end{align*}
Then there exists an $\bfx^\ro$ such that
\begin{align*}
    \Pr_\bfX[\bfX = \bfx^\ro | \RO = \ro, (\bfX \in C \land \RO(\bfX) = \mathbf{0})] &\geq 2^{-\minent(\secp)} > 0
\end{align*}
Furthermore, this $\bfx^\ro$ is correct ($\bfx^\ro \in C \land \ro(\bfx^\ro) = \mathbf{0}$) because otherwise, $\Pr_\bfX[\bfX = \bfx^\ro | \RO = \ro, (\bfX \in C \land \RO(\bfX) = \mathbf{0})] = 0$. Then
\begin{align*}
    \frac{\Pr_\bfX\left[\bfX = \bfx^\ro | \RO = \ro\right]}{\delta(\secp)} &\geq \Pr_\bfX[\bfX = \bfx^\ro | \RO = \ro, (\bfX \in C \land \RO(\bfX) = \mathbf{0})]\\
    &\geq 2^{-\minent(\secp)}\\
    \Pr_\bfX\left[\bfX = \bfx^\ro | \RO = \ro\right] &\geq \delta(\secp) \cdot 2^{-\minent(\secp)}\\
    &= p(\secp)
\end{align*}

In summary, if $\text{Event}_1$ occurs, then there exists an $\bfx^\ro$ such that:
\begin{itemize}
    \item $\bfx^\ro$ is correct: $\bfx^\ro \in C \land \ro(\bfx^\ro) = \mathbf{0}$, and 
    \item $\Pr_\bfX\left[\bfX = \bfx^\ro | \RO = \ro\right] \geq p(\secp)$.
\end{itemize}

Furthermore, since this $\bfx^\ro$ is correct, then $\text{Event}_2$ implies that this $\bfx^\ro$-value also satisfies the following condition:
\begin{itemize}
    \item There exists \textit{no} set $I \subseteq [n]$ of size $\geq n - s$ such that for all $i \in I$,
        \[w_{i, x^\ro_i}^{\leq D}(\ro) \geq \frac{p^2}{16 Q n}\]
\end{itemize}
This shows that $(\ro, \bfx^\ro)$ satisfy all the conditions of $\text{Event}_3$.
\end{proof}

\begin{lemma}
    $\Pr_\RO[\text{Event}_2]$ is overwhelming.
\end{lemma}
\begin{proof}
\Cref{thm:2-query-security} says that $\Pr_\RO[\text{Event}_2]$ is overwhelming as long as $p(\secp) = \omega\left(2^{-\frac{1}{2} \cdot \left(\secp^{c/2}\right)}\right)$. So it suffices to prove that $p$ satisfies the condition of \cref{thm:2-query-security}.

\begin{claim}
    $p(\secp) = \omega\left(2^{-\frac{1}{2} \cdot \left(\secp^{c/2}\right)}\right)$.
\end{claim}
\begin{proof}
    $\delta(\secp)$ is inverse-polynomial, so there exists a $d > 0$ such that
    \begin{align*}
        \delta(\secp) &= \secp^{-d} = 2^{- d \cdot \log \secp}\\
        p(\secp) &= 2^{- d \cdot \log \secp} \cdot 2^{-\minent(\secp)}\\
        &= 2^{-\left[d \cdot \log \secp + \minent(\secp)\right]}
    \end{align*}
    Next, note that 
    \begin{align*}
        \minent(\secp) &= o(\secp^{c/2})\\
        d \cdot \log \secp &= o(\secp^{c/2})\\
        \secp^{c/2} &= \omega(1)
    \end{align*}
    
    Then for any constant $\varepsilon > 0$, and for sufficiently large $\secp$,
    \begin{align*}
        \minent(\secp) &\leq \frac{1}{8} \cdot \secp^{c/2}\\
        d \cdot \log \secp &\leq \frac{1}{8} \cdot \secp^{c/2}\\
        4 \cdot \log \varepsilon &\leq \secp^{c/2}
    \end{align*}
    In this case,
    \begin{align*}
        \log \varepsilon &\leq \frac{1}{4} \cdot \secp^{c/2}\\
        -\frac{1}{4} \cdot \secp^{c/2} &\leq -\log \varepsilon\\\\
        d \cdot \log \secp + \minent(\secp) &\leq \frac{1}{4} \cdot \secp^{c/2}\\
        &= - \frac{1}{4} \cdot \secp^{c/2} + \frac{1}{2} \cdot \secp^{c/2}\\
        &\leq - \log \varepsilon + \frac{1}{2} \cdot \secp^{c/2}\\
        -\left[d \cdot \log \secp + \minent(\secp)\right] &\geq \log \varepsilon - \frac{1}{2} \cdot \secp^{c/2}\\
        p(\secp) &\geq 2^{\log \varepsilon - \frac{1}{2} \cdot \secp^{c/2}}\\
        &= \varepsilon \cdot 2^{- \frac{1}{2} \cdot \secp^{c/2}}
    \end{align*}

    Therefore, $p(\secp) = \omega\left(2^{- \frac{1}{2} \cdot \secp^{c/2}}\right)$
\end{proof}

\end{proof}

\begin{lemma}
    There is some non-negligible function $\text{non-negl}(\secp)$ such that $\Pr_{\RO, \bfX}[\text{Event}_3] \geq \text{non-negl}(\secp) \cdot p(\secp)$.
\end{lemma}
\begin{proof}
    We already know that $\Pr_{\RO}[\text{Event}_1]$ is non-negligible because the adversary breaks the certifiable min-entropy property, and we proved that $\Pr_{\RO}[\text{Event}_2]$ is overwhelming. This implies that $\Pr_{\RO}[\text{Event}_1 \land \text{Event}_2]$ is non-negligible.
    \begin{align*}
        \Pr_\RO[\text{Event}_1] &= \Pr_\RO[(\text{Event}_1 \land \text{Event}_2) \lor (\text{Event}_1 \land \lnot \text{Event}_2)]\\
        &= \Pr_\RO[\text{Event}_1 \land \text{Event}_2] + \Pr_\RO[\text{Event}_1 \land \lnot \text{Event}_2]\\
        &\leq \Pr_\RO[\text{Event}_1 \land \text{Event}_2] + \Pr_\RO[\lnot \text{Event}_2]\\
        &\leq \Pr_\RO[\text{Event}_1 \land \text{Event}_2] + \negl(\secp)\\
        \Pr_\RO[\text{Event}_1 \land \text{Event}_2] &\geq \Pr_\RO[\text{Event}_1] - \negl(\secp)\\
        &= \text{non-negl}(\secp) - \negl(\secp)\\
        &= \text{non-negl}'(\secp)
    \end{align*}

    Next, if $\text{Event}_1 \land \text{Event}_2$ occur for some $\ro$, then there exists some special $\bfx^\ro$ such that $(\ro, \bfx^\ro)$ satisfy the conditions of $\text{Event}_3$. Furthermore, $\text{Event}_3$ guarantees that:
    \[\Pr_\bfX\left[\bfX = \bfx^\ro | \RO = \ro\right] \geq p(\secp)\]
    
    One way to satisfy the conditions of $\text{Event}_3$ is to sample an $\ro$ that satisfies the conditions of $\text{Event}_1 \land \text{Event}_2$ and then to sample the particular $\bfx^\ro$ value that is guaranteed to satisfy $\text{Event}_3$. The probability of sampling values $(\ro, \bfx^\ro)$ of this form is
    \[\geq \Pr_\RO[\text{Event}_1 \land \text{Event}_2] \cdot p(\secp) \geq \text{non-negl}(\secp) \cdot p(\secp)\]

    Then 
    \[\Pr_{\RO, \bfX}[\text{Event}_3] \geq \text{non-negl}(\secp) \cdot p(\secp)\]
\end{proof}

\begin{claim}
    $p(\secp)^{-1} \cdot 2^{-(s(\secp)-1)}$ is non-negligible.
\end{claim}
\begin{proof}
\Cref{thm:must-heavy-query} says that \[\Pr_{\RO, \bfX}[\text{Event}_3] \leq 2^{-(s(\secp)-1)}\] for all $\secp$.
This implies that there is some non-negligible function $\text{non-negl}(\secp)$ such that for all $\secp$,
\begin{align*}
    \text{non-negl}(\secp) \cdot p(\secp) &\leq \Pr_{\RO, \bfX}[\text{Event}_3] \leq 2^{-(s(\secp)-1)}
\end{align*}
For sufficiently large $\secp$, $p(\secp) > 0$ because $p(\secp) = \omega\left(2^{-\frac{1}{2} \cdot \left(\secp^{c/2}\right)}\right)$. Then for sufficiently large $\secp$,
\begin{align*}
    \text{non-negl}(\secp) &\leq p(\secp)^{-1} \cdot 2^{-(s(\secp)-1)}
\end{align*}
If $p(\secp)^{-1} \cdot 2^{-(s(\secp)-1)}$ were negligible, then $\text{non-negl}(\secp)$ would be negligible as well, which is impossible. Therefore, $p(\secp)^{-1} \cdot 2^{-(s(\secp)-1)}$ is non-negligible.
\end{proof}

On the other hand, we can show that $p(\secp)^{-1} \cdot 2^{-(s(\secp)-1)}$ is negligible.
\begin{claim}
    $p(\secp)^{-1} \cdot 2^{-(s(\secp)-1)}$ is negligible.
\end{claim}
\begin{proof}
    We have that 
    \begin{align*}
        s(\secp) &= \lfloor \zeta n \rfloor\\
        n &= 2^{\lfloor \log \secp \rfloor} - 1
    \end{align*}
    \Cref{sec:prelim-YZ-protocol} defines $\zeta = \Omega(1)$, so there exists a constant $k > 0$ such that for all $\secp$,
    \begin{align*}
        \zeta(\secp) &\geq k
    \end{align*}
    
    Next,
    \begin{align*}
        n &\geq 2^{\log \secp - 1} - 1 = \frac{\secp}{2} - 1\\
        s(\secp) &= \lfloor \zeta n \rfloor\\
        &\geq \zeta n - 1\\
        s(\secp) - 1 &\geq \zeta n - 2
    \end{align*}
    Then for sufficiently large $\secp$,
    \begin{align*}
        n &> \frac{\secp}{3} + \frac{2}{k} > 0\\
        s(\secp) - 1 &\geq \zeta \cdot \left(\frac{\secp}{3} + \frac{2}{k}\right) - 2\\
        &\geq k \cdot \left(\frac{\secp}{3} + \frac{2}{k}\right) - 2\\
        &= \frac{k}{3} \cdot \secp + 2 - 2 = \frac{k}{3} \cdot \secp\\
        -(s(\secp)-1) &\leq -\frac{k}{3} \cdot \secp\\
        2^{-(s(\secp)-1)} &\leq 2^{-\frac{k}{3} \cdot \secp}
    \end{align*}
    Additionally, since $p(\secp) = \omega\left(2^{-\frac{1}{2} \cdot \left(\secp^{c/2}\right)}\right)$, we have that for sufficiently large $\secp$,
    \begin{align*}
        p(\secp) &> 2^{-\frac{1}{2} \cdot \left(\secp^{c/2}\right)} > 0\\
        p(\secp)^{-1} &< 2^{\frac{1}{2} \cdot \left(\secp^{c/2}\right)}\\
        p(\secp)^{-1} \cdot 2^{-(s(\secp)-1)} &\leq 2^{\frac{1}{2} \cdot \left(\secp^{c/2}\right)} \cdot 2^{-\frac{k}{3} \cdot \secp}\\
        &\leq 2^{\frac{1}{2} \cdot \left(\secp^{c/2}\right) - \frac{k}{3} \cdot \secp}
    \end{align*}
    Finally, since $0 < c < 1$, then for sufficiently large $\secp$,
    \begin{align*}
        \frac{1}{2} \cdot \left(\secp^{c/2}\right) &\leq \frac{k}{12} \cdot \secp\\
        \frac{1}{2} \cdot \left(\secp^{c/2}\right) - \frac{k}{3} \cdot \secp &\leq \frac{k}{12} \cdot \secp  - \frac{k}{3} \cdot \secp = -\frac{k}{4} \cdot \secp\\
        2^{\frac{1}{2} \cdot \left(\secp^{c/2}\right) - \frac{k}{3} \cdot \secp} &\leq 2^{-\frac{k}{4} \cdot \secp} = \negl(\secp)\\
        p(\secp)^{-1} \cdot 2^{-(s(\secp)-1)} &\leq \negl(\secp)
    \end{align*}

    Therefore, $p(\secp)^{-1} \cdot 2^{-(s(\secp)-1)}$ is negligible.
\end{proof}

We have shown that $p(\secp)^{-1} \cdot 2^{-(s(\secp)-1)}$ is both non-negligible and negligible, which is a contradiction. Therefore, our initial assumption was false, and in fact, \cref{thm:CR-for-depth-bounded-adversaries} is true.

\section{AI Disclosure}
The authors used AI to improve clarity and polish writing in Sections 1 and 2. All techniques in this work were developed by human trial-and-error, without AI assistance. 

\section{Acknowledgments}
AT was supported by NSF CAREER award CCF-2145474.
Parts of this work were done while the authors were visiting the Simons Institute for
the Theory of Computing, Berkeley.

\bibliographystyle{alpha}
\bibliography{references.bib}

@phdthesis{Col06,
  author    = {Roger Colbeck},
  title     = {Quantum and Relativistic Protocols for Secure Multi-Party Computation},
  school    = {University of Cambridge},
  year      = {2006},
  note      = {Available as arXiv:0911.3814}
}

@article{PAM+10,
  author    = {Stefano Pironio and Antonio Ac{\'i}n and Serge Massar and
               Antoine {Boyer de la Giroday} and Dzmitry N. Matsukevich and
               Peter Maunz and Steven Olmschenk and David Hayes and Le Luo and
               T. Andrew Manning and Christopher Monroe},
  title     = {Random Numbers Certified by {Bell}'s Theorem},
  journal   = {Nature},
  volume    = {464},
  number    = {7291},
  pages     = {1021--1024},
  year      = {2010},
  doi       = {10.1038/nature09008}
}

@inproceedings{VV12,
  author    = {Umesh V. Vazirani and Thomas Vidick},
  title     = {Certifiable Quantum Dice: Or, True Random Number Generation
               Secure Against Quantum Adversaries},
  booktitle = {Proceedings of the 44th Annual {ACM} Symposium on Theory of
               Computing ({STOC} '12)},
  pages     = {61--76},
  publisher = {ACM},
  year      = {2012},
  doi       = {10.1145/2213977.2213984}
}

@article{BCM+21,
  author    = {Zvika Brakerski and Paul Christiano and Urmila Mahadev and
               Umesh V. Vazirani and Thomas Vidick},
  title     = {A Cryptographic Test of Quantumness and Certifiable Randomness
               from a Single Quantum Device},
  journal   = {Journal of the {ACM}},
  volume    = {68},
  number    = {5},
  pages     = {31:1--31:47},
  year      = {2021},
  doi       = {10.1145/3441309}
}

@inproceedings{BDF+11,
  author    = {Dan Boneh and {\"O}zg{\"u}r Dagdelen and Marc Fischlin and
               Anja Lehmann and Christian Schaffner and Mark Zhandry},
  title     = {Random Oracles in a Quantum World},
  booktitle = {Advances in Cryptology -- {ASIACRYPT} 2011},
  series    = {Lecture Notes in Computer Science},
  volume    = {7073},
  pages     = {41--69},
  publisher = {Springer},
  year      = {2011},
  doi       = {10.1007/978-3-642-25385-0_3}
}

@inproceedings{AH23,
  author       = {Scott Aaronson and
                  Shih{-}Han Hung},
  editor       = {Barna Saha and
                  Rocco A. Servedio},
  title        = {Certified Randomness from Quantum Supremacy},
  booktitle    = {Proceedings of the 55th Annual {ACM} Symposium on Theory of Computing,
                  {STOC} 2023, Orlando, FL, USA, June 20-23, 2023},
  pages        = {933--944},
  publisher    = {{ACM}},
  year         = {2023},
  url          = {https://doi.org/10.1145/3564246.3585145},
  doi          = {10.1145/3564246.3585145},
  bibsource    = {dblp computer science bibliography, https://dblp.org}
}

@article{YZ24,
   title={Verifiable Quantum Advantage without Structure},
   volume={71},
   ISSN={1557-735X},
   url={http://dx.doi.org/10.1145/3658665},
   DOI={10.1145/3658665},
   number={3},
   journal={Journal of the ACM},
   publisher={Association for Computing Machinery (ACM)},
   author={Yamakawa, Takashi and Zhandry, Mark},
   year={2024},
   month=jun, pages={1–50} }

@article{AA14,
  author       = {Scott Aaronson and
                  Andris Ambainis},
  title        = {The Need for Structure in Quantum Speedups},
  journal      = {Theory Comput.},
  volume       = {10},
  pages        = {133--166},
  year         = {2014}
}

@article{Vaz98,
    author = {Ekert, A. and Jozsa, R. and Penrose, R.  and Vazirani, Umesh},
    title = {On the power of quantum computation},
    journal = {Philosophical Transactions of the Royal Society A: Mathematical, Physical and Engineering Sciences},
    volume = {356},
    number = {1743},
    pages = {1759-1768},
    year = {1998},
    month = {08},
    issn = {1364-503X},
    doi = {10.1098/rsta.1998.0247},
    url = {https://doi.org/10.1098/rsta.1998.0247},
    eprint = {https://royalsocietypublishing.org/rsta/article-pdf/356/1743/1759/1450345/rsta.1998.0247.pdf},
}

@article{BBBV97,
   title={Strengths and Weaknesses of Quantum Computing},
   volume={26},
   ISSN={1095-7111},
   url={http://dx.doi.org/10.1137/S0097539796300933},
   DOI={10.1137/s0097539796300933},
   number={5},
   journal={SIAM Journal on Computing},
   publisher={Society for Industrial & Applied Mathematics (SIAM)},
   author={Bennett, Charles H. and Bernstein, Ethan and Brassard, Gilles and Vazirani, Umesh},
   year={1997},
   month=Oct, pages={1510–1523} }

@inproceedings{BSW22,
  author       = {Nikhil Bansal and
                  Makrand Sinha and
                  Ronald de Wolf},
  title        = {Influence in Completely Bounded Block-Multilinear Forms and Classical
                  Simulation of Quantum Algorithms},
  booktitle    = {{CCC}},
  series       = {LIPIcs},
  volume       = {234},
  pages        = {28:1--28:21},
  publisher    = {Schloss Dagstuhl - Leibniz-Zentrum f{\"{u}}r Informatik},
  year         = {2022}
}

@article{Mon12,
  title={Some applications of hypercontractive inequalities in quantum information theory},
  author={Montanaro, Ashley},
  journal={Journal of Mathematical Physics},
  volume={53},
  number={12},
  year={2012},
  publisher={AIP Publishing}
}

@inproceedings{Bhattacharya25,
  author       = {Sreejata Kishor Bhattacharya},
  title        = {Random Restrictions of Bounded Low Degree Polynomials Are Juntas},
  booktitle    = {{ITCS}},
  series       = {LIPIcs},
  volume       = {325},
  pages        = {17:1--17:21},
  publisher    = {Schloss Dagstuhl - Leibniz-Zentrum f{\"{u}}r Informatik},
  year         = {2025}
}

@article{BBFGT26,
  author       = {Roozbeh Bassirian and
                  Adam Bouland and
                  Bill Fefferman and
                  Sam Gunn and
                  Avishay Tal},
  title        = {On Certified Randomness from Fourier Sampling or Random Circuit Sampling},
  journal      = {Quantum},
  volume       = {10},
  pages        = {2002},
  year         = {2026}
}

@article{MVV22,
  author       = {Urmila Mahadev and
                  Umesh V. Vazirani and
                  Thomas Vidick},
  title        = {Efficient Certifiable Randomness from a Single Quantum Device},
  journal      = {CoRR},
  volume       = {abs/2204.11353},
  year         = {2022}
}

\end{document}